\documentclass[11pt,a4paper]{scrartcl}

\usepackage[ansinew]{inputenc}
\usepackage[T1]{fontenc}

\usepackage[pdftex]{graphicx}
\usepackage{epstopdf}

\usepackage[format=plain,justification=centering,singlelinecheck=true]{caption}
\usepackage{float}
\usepackage{latexsym}
\usepackage{amsmath,amssymb,amsthm,tabu,amsfonts}
\usepackage{makeidx} 
\usepackage[numbers]{natbib}
\usepackage{enumitem}
\usepackage{pdfpages}
\usepackage{color}
\usepackage{todonotes}
\usepackage{colortbl}
\usepackage{xcolor}
\usepackage{pifont}
\usepackage{booktabs}
\usepackage{tocloft} 
\usepackage{hyperref}
\usepackage[hang,flushmargin,multiple,ragged]{footmisc}

\usepackage{caption}
\usepackage[font=normalsize]{subcaption}

\allowdisplaybreaks
\newtheorem{theorem}{Theorem}
\newtheorem{definition}{Definition} 
\newtheorem{remark}{Remark}	   
\newtheorem{assumption}{Assumption}
\newtheorem{example}{Example}
\newtheorem{Lemma}{Lemma}
                  
\numberwithin{equation}{section}

\newcommand{\vech}{\operatorname{vech}}
\newcommand{\Cov}{\operatorname{Cov}}

\newcommand{\E}{\operatorname{E}}
\newcommand{\diag}{\operatorname{diag}}

\begin{document}

\title{Detecting changes in the covariance structure of functional time series with application to fMRI data}

\author{Christina Stoehr\footnote{\, Ruhr-Universit\"at Bochum, Department of Mathematics,  Bochum, Germany; christina.stoehr@ruhr-uni-bochum.de}\and John A D Aston\footnote{\, University of Cambridge,  Statistical Laboratory, Cambridge, UK} \and Claudia Kirch\footnote{\, Otto-von-Guericke University Magdeburg, Institute for Mathematical Stochastics,  Magdeburg, Germany}\textsuperscript{, }\footnote{\, Center for Behavioral Brain Sciences (CBBS), Magdeburg, Germany}}

\maketitle
\begin{center}
\begin{minipage}{0.8\textwidth}
\begin{center}\textbf{Abstract}\end{center}
Functional magnetic resonance imaging (fMRI) data provides information concerning activity in the brain and in particular the interactions between brain regions. Resting state fMRI data is widely used for inferring connectivities in the brain which are not due to external factors. As such analyzes strongly rely on stationarity, change point procedures can be applied in order to detect possible deviations from this crucial assumption. In this paper, we model fMRI data as functional time series and develop tools for the detection of deviations from covariance stationarity via change point alternatives. We propose a nonparametric procedure which is based on dimension reduction techniques. However, as the projection of the functional time series on a finite and rather low-dimensional subspace involves the risk of missing changes which are orthogonal to the projection space, we also consider two test statistics which take the full functional structure into account. The proposed methods are compared in a simulation study and applied to more than 100 resting state fMRI data sets.
\end{minipage}\\
\end{center}

\textbf{Keywords:} Change point analysis, covariance change, functional data, resting state fMRI, functional time series, dimension reduction

\section{Introduction}
Functional Magnetic Resonance Imaging (fMRI) is a widely used technique to capture brain activity. An fMRI dataset consists of a sequence of three-dimensional images related to the contrast of oxygenated and deoxygenated hemoglobin, the so called BOLD signal, that are recorded every few seconds. fMRI facilitates a noninvasive real time functional brain mapping with a high spatial resolution and thus yields large amounts of data requiring the development of appropriate statistical methodologies. fMRI scans can be obtained related to a task or in a resting state where the person is told to go through the scanning procedure without thinking of anything while not falling asleep. Resting state data is used to analyze brain activities excluding external factors where the examination of the covariance structure between brain regions is of particular interest as it is associated with neural connectivity. Such analyses strongly rely on the assumption that resting state data is first and second order stationary. This assumption is by no means guaranteed as it might happen, for example, that during the scan the person suddenly remembers something such that the mean activities deviate from their resting state baseline in some areas of the brain. If such a scan is then used for analyzing connectivities without taking a possible change into account, the results will be contaminated by the mean change leading to wrong conclusions. Therefore, in \cite{fda2}, Aston and Kirch developed testing procedures to detect deviations from mean stationarity. However, it is not only deviations from mean stationarity but also deviations from covariance stationarity that will contaminate the analysis and ultimately the conclusions. Therefore, in this paper, we develop tools to test for deviations from covariance stationarity in fMRI data which will be modeled as functional time series. This means that each observation of the time series, in this case each 3-d image of the brain, can be viewed as a function.
Indeed, taking into consideration that the brain works as a single unit with spatial dependencies, it is a natural approach to model each image as a discretized observation of a functional response. In contrast, a voxelwise approach requires a difficult adaption for multiple testing and may miss signals that are very small in any voxel but considerably large if information across voxels is used. Dependencies in time, i.e. between subsequent images, which are also present in fMRI data, can be captured by a time series structure. Lifting the multivariate observations to a functional space makes them mathematically easier to handle as one can exploit functional properties, such as smoothness, making use of many well established statistical techniques.\\

The statistical analysis of functional data is currently a rapidly progressing field of research as an increasing number of applications provides data which can be modeled in such a way. The methodology developed in this paper is widely applicable beyond the considered application of fMRI data, hence also of independent interest in functional data analysis in general. We adapt a nonparametric approach where we tackle the problem by means of a change point procedure without assuming any parametric spatial or temporal correlation structure. Such nonparametric methods become more and more refined in the analysis of functional data (cf. \cite{Fer} and \cite{Inffd}). Nonparametric tests for at most one change (AMOC) in the mean function have been considered for independent observations in \cite{AueFd} and \cite{Berk} as well as for weakly dependent data in \cite{HK}. Aston and Kirch \cite{fda} extend these results to a more general class of dependency structures and also consider epidemic changes where the mean function returns to its original state after some time.\\
The analysis of functional connectivity data is a very active field of research in neuroimaging. The detection of change points in the observed data without assuming the specifications of the experiment to be known is of particular interest. In this context, Cribben et al. \cite{Crib1} propose a data-driven approach, the so called Dynamic Connectivity Regression (DCR), for detecting changes in the functional connectivity between a set of brain regions and estimate a connectivity graph for each temporal interval between the change points. They use resampling methods in order to decide whether a change is significant. With a view to single-subject data, DCR is further developed in \cite{Crib2}.
In this paper we develop statistical procedures for the detection of deviations from covariance stationarity in functional time series that can be applied to fMRI data without being restricted to predefined regions of interest.\\

The paper is organized as follows: In Section \ref{testing} we propose a procedure based on dimension reduction techniques such as principal component analysis, to detect deviations from covariance stationarity. The test statistics and their asymptotic behavior are investigated in Section \ref{teststat}. The proposed procedures require the estimation of the long-run covariance which is statistically unstable. Using a missspecified estimator is a possible solution but leads to an unknown limit distribution such that resampling procedures, as described in Section \ref{resampling}, are unavoidable. Alternative test statistics which take the full functional structure into account without reducing the dimension are discussed in Section \ref{alttest}. The different procedures proposed in this paper are compared in a simulation study in Section \ref{sim}. The application to fMRI data is presented in Section \ref{appl}. Additional technical details, proofs and further results from the data analysis are given in the supplementary material.
\section{Testing for changes in the covariance structure of functional data}\label{testing}
We assume that the observations are obtained from a functional time series with the respective mean function being constant over time, i.e.
$$X_t(s)=\mu(s)+Y_t(s),\quad 1\leq t\leq n,$$
where $t$ denotes the time point and $s$ a spatial coordinate in a compact set $\mathcal{Z}$. The constant mean function is given by $\mu(\cdot)$ while the random fluctuations are represented by $Y_t(\cdot)$ with $\E(Y_t(s))=0$ which is not necessarily stationary but can have a time-dependent covariance structure as detailed in Section \ref{cpmodel}. $\mu(\cdot)$ as well as all elements of $\{Y_t(\cdot):1\leq t\leq n\}$ are assumed to be square integrable on $\mathcal{Z}$.
The mean stationarity can be checked previously as described in \cite{fda2}.\\
The covariance structure of a functional time series is determined by the covariance operator respectively the covariance kernel as given in the following definition.
\begin{definition}
Let $\{X_t(\cdot):1\leq t\leq n\}\in \mathcal{L}^2(\mathcal{Z})$ be a functional time series, where $\mathcal{Z}$ is a compact set. The square integrable covariance operator $C_t:\mathcal{L}^2(\mathcal{Z})\mapsto \mathcal{L}^2(\mathcal{Z})$ is defined by $$C_t(z)=\int_{\mathcal{Z}}c_t(\cdot,s)z(s)ds,$$ where $c_t(u,s)=\Cov\left(X_t(u),X_t(s)\right)$ is the covariance kernel of $X_t(\cdot).$
\end{definition}

\subsection{Change point model}\label{cpmodel}
First consider the at most one change (AMOC) alternative given by
\begin{align}\label{ysplit}
Y_t(s)=Y_t^{(1)}(s)1_{\{1\leq t\leq \theta n\}}+Y_t^{(2)}(s)1_{\{\theta n<t\leq n\}},\quad 1\leq t\leq n
\end{align}
with $\Cov\left(Y_t^{(1)}(u),Y_t^{(1)}(s)\right)=c(u,s)$ and $\Cov\left(Y_t^{(2)}(u),Y_t^{(2)}(s)\right)=c(u,s)+\delta(u,s)$
for some $0<\theta<1$ and $c(u,s),\delta(u,s)\in\mathcal{L}^2(\mathcal{Z}\times\mathcal{Z})$. 
According to this model, the covariance change occurs at the unknown time point $[\theta n]$. The covariance kernel $c(u,s)$ before the change as well as the change in covariance $\delta(u,s)\neq 0$ are both unknown. 
\begin{assumption}\label{ass.erg}
Assume that for $\{Y_t(\cdot)\}$ as in \eqref{ysplit} it holds
\begin{itemize}
\item[(i)] 
$\{Y_t^{(1)}(\cdot)\}\in \mathcal{L}^2(\mathcal{Z})$ with $$\E Y^{(1)}_1(s)=0\quad\mbox{and}\quad\E\|Y^{(1)}_1(\cdot)\|^4=\int \E\left[\left(Y^{(1)}_1(s)\right)^4\right]ds<\infty$$
is $L_m^4-$approximable \cite{HK} and hence, in particular, stationary and ergodic.
\item[(ii)]
$\{Y_t^{(2)}(\cdot)\}\in \mathcal{L}^2(\mathcal{Z})$ is ergodic with $$\E Y^{(2)}_1(s)=0\quad\mbox{and}\quad\E\|Y^{(2)}_1(\cdot)\|^2=\int \E\left[\left(Y^{(2)}_1(s)\right)^2\right]ds<\infty.$$
\end{itemize}
\end{assumption}
As we do not assume $Y_t^{(2)}$ to be stationary, the time series after the change is allowed to have starting values from a different distribution. $L_m^4-$approximability is a nonparametric concept of dependence which provides the necessary mathematical tools for the proofs and is satisfied for a large class of time series. Full details can be found in \cite{HK}.\\

Testing for covariance stationarity against the AMOC alternative can be described by the following hypotheses:
\begin{align*} 
H_0:\theta=1\quad
\mbox{against}\quad H_1:\quad&0<\theta<1.
\end{align*}
In order to obtain a test for more general alternatives of nonstationarities in the covariance, we consider the following epidemic alternative:
$$Y_t(s)=Y_t^{(1)}(s)1_{\{1\leq t\leq \theta_1 n\}}+Y_t^{(2)}(s)1_{\{\theta_1 n<t\leq \theta_2 n\}}+Y_t^{(1)}(s)1_{\{\theta_2 n< t\leq n\}},\quad 1\leq t\leq n$$
with $\Cov\left(Y_t^{(1)}(u),Y_t^{(1)}(s)\right)=c(u,s)$ and $\Cov\left(Y_t^{(2)}(u),Y_t^{(2)}(s)\right)=c(u,s)+\delta(u,s)$ for some $0<\theta_1<\theta_2<1$. It would also be possible to allow for contaminated starting values in the time series after the change. This alternative can be viewed as a better approximation to the expected kind of deviation from covariance stationarity.
\subsection{Dimension reduction techniques}\label{dimred}
A common approach in functional data analysis is the transition to a multivariate setting by projecting the data into a $d$-dimensional space spanned by an orthonormal basis $\{v_k(\cdot):k=1,\ldots,d\}$. In this case, the projection scores are obtained by
\begin{align}\label{scores}
\langle X_t,v_l\rangle =\int X_t(s)v_l(s)ds,\quad t=1,\ldots,n, l=1,\ldots,d.
\end{align}
As we aim at assessing the above functional testing problem by applying a multivariate testing procedure to the projection scores we first need to verify if a change in the covariance structure of the observed functional time series implies a change in the covariance of the scores. To this end, we observe
\begin{align*}
&\Cov(\langle X_t,v_{l_1}\rangle,\langle X_t,v_{l_2}\rangle)\\
=&\int\int c(u,s)v_{l_1}(u)v_{l_2}(s)du\,ds+1_{\{\theta n<t\leq n\}}\int\int\delta(u,s)v_{l_1}(u)v_{l_2}(s)du\,ds.
\end{align*}
Thus, a necessary condition for the covariance change to be visible in the projection scores is
\begin{equation}\label{changedet}
\int\int\delta(u,s)v_{l_1}(u)v_{l_2}(s)du\,ds\neq 0\quad\mbox{for some }l_1,l_2=1,\dots,d.
\end{equation}
In contrast to other applications we do not require the dimension reduction  technique to explain a large amount of the variation of the data but to yield a good signal-to-noise ratio where the signal is determined by $\int\int\delta(u,s)v_{l_1}(u)v_{l_2}(s)du\,ds$.
\paragraph*{Principal component analysis}
Principal component analysis (PCA) is a widely used data driven dimension reduction technique which projects the functional data on the subspace spanned by the first $d$ principal components explaining the most variance of any subspace of size $d$.
Let $\{\lambda_l:l\geq 1\}$ be the non-negative decreasing sequence of eigenvalues of the covariance operator and
$\{v_l(\cdot):l\geq 1\}$ a set of corresponding orthonormal eigenfunctions  defined by
$$\int c(u,s)v_l(s)ds=\lambda_lv_l(u),\quad l=1,2,\ldots;u\in\mathcal{Z}.$$
By Mercer's Lemma, see Lemma 1.3 in \cite{Bosq}, the covariance kernel can be expressed as
$$c(u,s)=\sum_{l=1}^{\infty}\lambda_kv_l(u)v_l(s)$$
and the  Karhunen-Lo\`{e}ve expansion, see Theorem 1.5 in \cite{Bosq}, yields
\begin{align}\label{KL}
X_t(s)-\mu(s)=\sum_{l=1}^{\infty}\eta_{t,l}v_l(s),
\end{align}
where the scores $\{\eta_{t,l}:l=1,2,\ldots\}$ given by $\eta_{t,l}=\int \left(X_t(s)-\mu(s)\right)v_l(s)ds$ are uncorrelated and centered with variance $\lambda_l$. As the covariance kernel is unkown, PCA is usually conducted based on the empirical covariance function
$$\hat{c}_n(u,s)=\frac1n\sum_{t=1}^n\left(X_t(u)-\overline{X}_n(u)\right)\left(X_t(s)-\overline{X}_n(s)\right),$$
where $\overline{X}_n(s)=\frac 1n\sum_{t=1}^nX_t(s).$ Under the null hypothesis, the empirical covariance function estimates the actual covariance kernel $c(u,s)$ whereas under the alternative it converges to a contaminated limit $k(u,s)$ as stated in (\ref{cconv}) in the supplementary material. As projection basis we determine the eigenfunctions $\{\hat{v}_l(\cdot):l=1,\ldots,d\}$ of $\hat{c}_n$ belonging to the $d$ largest eigenvalues and obtain the projection scores by $$\hat{\eta}_{t,l}=\int \left(X_t(s)-\overline{X}_n(s)\right)\hat{v}_l(s)ds=\langle X_t,\hat{v}_l\rangle -\overline{\langle X,\hat{v}_l\rangle}_n,\quad t=1,\ldots,n, l=1,\ldots,d$$
with $\overline{\langle X,\hat{v}_l\rangle}_n=\frac 1n\sum_{t=1}^n\langle X_t,\hat{v}_l\rangle.$ For more details on functional principal component analysis, in particular for the consistency of the empirical eigenvalues and eigenfunctions, see, for example, \cite{Inffd}.
\paragraph*{Separable covariance structure}
As fMRI data is collected voxelwise ($\sim M:=10^5$ voxels), using the empirical covariance function requires the calculation and storage of an $M\times M$-dimensional matrix in addition to the respective eigenanalysis. While this is computationally infeasible, one can show that there is a one-to-one correspondence between the eigenvalues and eigenvectors of the spatial covariance matrix ($M\times M$) and that of the time domain ($n\times n$). As $M\gg n$ the eigenanalysis in the time domain requires less computational effort. However, this relationship also reveals that the number of nonzero eigenvalues is limited by the sample size and hence indicates a considerable loss of precision when using the nonparametric covariance estimator. Based on those considerations, Aston and Kirch \cite{fda2} suggest to use a separable covariance structure in the estimation procedure given by
$$c\left((u_1,u_2,u_3),(s_1,s_2,s_3)\right)= c_1\left(u_1,s_1\right)c_2\left(u_2,s_2\right)c_3\left(u_3,s_3\right).$$
In this work, we adopt this approach of estimating the covariance matrix separately in each direction ($64\times 64$ resp. $64\times 64$  resp. $33\times 33$) and calculate the respective eigenvalues and eigenfunctions.
The  projection basis can then be obtained by the tensor product of the first $d$ eigenfunctions of each direction. Even if the actual covariance structure is not separable we obtain a valid projection such that the proposed dimension reduction can be applied for our purposes. While this is an obvious simplification, most smoothing techniques in fMRI make use of tensor based formulations leading to very similar implicit assumptions.

\subsection{Test statistic and statistical properties}\label{teststat}
We assess the functional testing problem by testing for a change in the covariance structure of the d-dimensional estimated score vectors. As proposed by Aue et al. in \cite{Aue} we construct the test statistic based on the following version of the traditional CUSUM-statistic for the AMOC alternative:
\begin{align}\label{sk1}
S_k=\frac{1}{\sqrt{n}}\left(\sum_{t=1}^k\vech[\hat{\eta}_t\hat{\eta}_t^T]-\frac kn\sum_{t=1}^n\vech[\hat{\eta}_t\hat{\eta}_t^T]\right).
\end{align}
We consider the test statistic
\begin{equation*}
\Omega_n=\frac 1n\sum_{k=1}^n S_k^T\hat{\Sigma}_n^{-1}S_k,
\end{equation*}
where $\hat{\Sigma}_n$ is an estimator for the long-run covariance $\Sigma_0=\sum_{t\in\mathbb{Z}}\Cov\left(\vech[\eta_0\eta_0^T],\vech[\eta_t\eta_t^T]\right)$ under $H_0$. We assume that $\hat{\Sigma}_n$ is consistent under the null hypothesis and
$$|\hat{\Sigma}_n-\Sigma_1|=o_p(1)\quad\mbox{under }H_1,$$
where $\Sigma_1$ is some positive-definite matrix which can differ from $\Sigma_0$. For the epidemic change point alternative we propose
the test statistic
\begin{equation*}
\Omega^{ep}_n=\frac 1n\sum_{1\leq k_1<k_2\leq n} S_{k_1,k_2}^T\hat{\Sigma}_n^{-1}S_{k_1,k_2}
\end{equation*}
with $S_{k_1,k_2}=S_{k_2}-S_{k_1}.$
While we consider sum-type test statistics throughout this paper the respective results for max-type test statistics obtained as the maximum over the quadratic forms $S_k^T\hat{\Sigma}_n^{-1}S_k$ resp. $S_{k_1,k_2}^T\hat{\Sigma}_n^{-1}S_{k_1,k_2}$ can be found in Section \ref{ap.max} in the supplementary material.
\paragraph*{Behavior under the null hypothesis}
We allow for a weak dependency structure by assuming the observed functional time series to be $L_m^4-$approximable. Effectively, this means the time series can be well approximated (in an $L_m^4-$sense) by an $m$-dependent one (see Definition 2.1 in \cite{HK}) - a property that many of the usual time series models possess.
\begin{theorem}\label{nullas}
Let Assumption \ref{ass.erg} (i) be satisfied. Additionally, we assume that the first $d+1$ eigenvalues of $c(u,s)$ are separated, i.e. $\lambda_1>\lambda_2>\ldots >\lambda_d>\lambda_{d+1}$. Then, the following asymptotics hold under the null hypothesis if $\hat{\Sigma}$ is a consistent estimator for the long-run covariance $\Sigma$.
\begin{align*}
\Omega_n\stackrel{\mathcal D}{\rightarrow}\sum_{l=1}^{\mathfrak{d}}\int_0^1 B_l^2(x)dx\quad\mbox{and}\quad\Omega_n^{ep}\stackrel{\mathcal D}{\rightarrow}\sum_{l=1}^{\mathfrak{d}}\int_0^1\int_{0}^y \left(B_l(y)-B_l(x)\right)^2dx\,dy,
\end{align*}
where \(\mathfrak{d}=d(d+1)/2\) and \((B_l(x):x\in [0,1],1\leq l\leq\mathfrak{d})\) are independent standard Brownian bridges.
\end{theorem}
Based on this result we can now determine the critical value as  $(1-\alpha)$-quantile of the respective limit distribution. This can be done by using Monte Carlo simulations. However, it is notoriously difficult to estimate the long-run covariance (see discussion in \cite{fda2}). In this case, i.e. if $\hat{\Sigma}$ is not consistent or the convergence too slow to be appropriate for small samples, the limit distributions in Theorem \ref{nullas} are no longer true.
\paragraph*{Behavior under the alternative hypothesis}
Condition (\ref{changedet}) is examined for two exemplary alternatives, where the projection basis is determined based on principal component analysis. The following Lemma states that, under the alternative, the empirical covariance function converges to a contaminated limit $k(u,s)$.
\begin{Lemma}\label{climit}
Under Assumption \ref{ass.erg} it holds
\begin{align*}
\int\int \left(\hat{c}_n(u,s)-k(u,s)\right)^2du\,ds=o_P(1),
\end{align*}
where $k(u,s)=c(u,s)+(1-\theta)\delta(u,s)$.
\end{Lemma}
\begin{example}[Change does not affect eigenfunctions]
We consider a covariance change that does not affect the eigenfunctions, i.e. the covariance kernel after the change has the same eigenfunctions $v_l(\cdot)$ as the covariance kernel before the change:
$$\int\left(c(u,s)+\delta(u,s)\right)v_l(s)ds=\tilde{\lambda}_lv_l(u),$$
where $v_l(\cdot)$ and $\lambda_l$ are the eigenfunctions and eigenvalues of $c(u,s)$ and $\tilde{\lambda}_l=\lambda_l+\delta_l$ with $\delta_l\neq 0$ for some l=1,\dots,d. \\
Condition (\ref{changedet}) is fulfilled as it holds (see Section \ref{proofs} in the supplementary material)
\begin{align*}
\int\int\delta(u,s)v_{l_1}(u)v_{l_2}(s)du\,ds=\begin{cases}
0,&{l_1}\neq {l_2}\\
\delta_l,& {l_1}={l_2}.
\end{cases}
\end{align*}
Assuming that the eigenvalues of $k(u,s)$ are separated, the change is still detectable if the eigendirections are estimated based on the empirical covariance function (see (\ref{ex1.del})).
\end{example}
\begin{example}[Additive noise term]\label{ex2}
In this example, a covariance change in the functional time series occurs due to an additive noise term in the scores of the first m leading eigendirections. More precisely, it holds $X_t(s)-\mu(s)=\sum_{l=1}^{\infty}\tilde{\eta}_{t,l}v_l(s)$ with
$$\tilde{\eta}_{t,l}=\eta_{t,l}+1_{\{\theta n<t\leq n,1\leq l\leq m\}}\epsilon_{t,l},$$
where $\epsilon_1,\dots,\epsilon_n$ with $\epsilon_t=(\epsilon_{t,1},\ldots,\epsilon_{t,m})$ are independent and identically distributed with mean 0 and $\Cov\left(\epsilon_{t,l_1},\epsilon_{t,l_2}\right)=\sigma_{l_1,l_2}$ and independent of $\eta$. In this setting, it holds (see Section \ref{proofs} in the supplementary material)
\begin{align*}
\int\int\delta(u,s)v_{l_1}(u)v_{l_2}(s)du\,ds&=\sigma_{l_1,l_2}
\end{align*}
for $l_1,l_2\in\{1,\ldots,m\}$. Hence, condition \eqref{changedet} is fulfilled if $\sigma_{l_1,l_2}\neq 0$ for some  $l_1,l_2\in\{1,\ldots,m\}$.
According to (\ref{ex2.est}) the change can be detected by projecting on the subspace spanned by the first $d$ eigendirections of the empirical covariance kernel  if $$\sum_{k,l=1}^{m}\sigma_{k,l}\left(\int v_{k}(u)\tilde{v}_{l_1}(u)du \int v_{l}(s)\tilde{v}_{l_2}(s)ds\right)\neq 0$$ for at least one pair $l_1,l_2\in\{1,\ldots,\min\{d,m\}\}$, where $\{\tilde{v}_l(\cdot):l\geq 1\}$ are the eigenfunctions of $k(u,s)$. 
\end{example}
\paragraph*{Estimation of the long-run covariance}
The estimation of the long-run covariance matrix is a challenging issue in change point analysis. In the case where $Y_j$ are independent under $H_0$ the long-run covariance reduces to the covariance, i.e.
\begin{align*}
\Sigma_0=\Cov\left(\vech[\eta_0\eta_0^T]\right)=\E\left(\vech[\eta_0\eta_0^T]\vech[\eta_0\eta_0^T]^T\right)-\E\left(\vech[\eta_0\eta_0^T]\right)\E\left(\vech[\eta_0\eta_0^T]\right)^T.
\end{align*}
The components of the scores are known to be uncorrelated. However, this does not necessarily imply a diagonal long-run covariance as, in general, the squared components are not uncorrelated. By additionally assuming that the scores are Gaussian we get a diagonal long-run covariance depending only on the eigenvalues of the covariance kernel which can be estimated by the eigenvalues of the estimated covariance kernel. More precisely, it holds (see Section \ref{proofs} in the supplementary material for proof):
\begin{align}\label{normalsig}
\Sigma=\Cov\left(\vech[\eta_0\eta_0^T]\right)=\diag(2\lambda_{1}^2,\lambda_{1}\lambda_{2},\ldots,2\lambda_{2}^2,\lambda_{2}\lambda_{3},\ldots,2\lambda_{d}^2).
\end{align}
However, when dealing with a time series structure and non-Gaussian structure one has to estimate the full long-run covariance. Usual estimators, such as the Bartlett estimator, lead to problems, in particular if the dimension is large compared to the sample size (see \cite{fda2}). Aston and Kirch \cite{fda2} conclude that the change point procedure becomes more stable and conservative if one only corrects for the long-run variance, i.e. the diagonal of the long-run covariance matrix. In our case, this approach leads to the following test statistic:
\begin{equation}\label{statadapt}
\tilde{\Omega}_n=\frac 1n\sum_{k=1}^n S_k^T\hat{D}_n^{-1}S_k,
\end{equation}
where $\hat{D}_n^{-1}$ is an estimator for the inverse of the diagonal matrix given by the diagonal elements of $\Sigma$. This test statistic is not pivotal in the sense that the asymptotic critical value depends on the unknown correlation structure. As a consequence, this approach requires resampling procedures. As detailed in the next section we apply a circular block bootstrap where we estimate the long-run variance of the bootstrap samples by the block sample variance given in (\ref{blockest}). The estimator  $\hat{D}_n$ for the test statistic has to be chosen carefully with respect to its interaction with the estimator used for the bootstrap statistic. We decide to estimate the long-run variance for the test statistic with the block estimator in \eqref{lrest} as, based on simulations, this seems to yield the most stable size in comparison to, for example, the flat-top kernel estimator introduced in \cite{Pol} with automatic bandwidth selection.
\subsection{Resampling procedures}\label{resampling}
The critical values of change point procedures are usually chosen based on the limit distribution of the test statistic under the null hypothesis. Resampling methods can be applied to get a better small sample performance but cannot be avoided if the limit distribution is non-pivotal and cannot be estimated otherwise, as is the case in our example. Previous work on resampling procedures for functional time series include \cite{Murry} for independent data and \cite{Hboot} as well as \cite{Hboot2} for dependent Hilbert space-valued random variables. Recently, in \cite{sievebs}, a sieve-type bootstrap procedure for functional time series based on a vector autoregressive representation of the scores has been introduced.\\
In order to prove the validity of a bootstrap procedure it has to be shown that, given the observations, the bootstrap test statistic has the same limit distribution as the actual test statistic under the null hypothesis and thus leads to the same asymptotic critical values. For a good power behavior under alternatives, it is important to take into account that the underlying observations may contain a change. Ideally, the respective limit distribution holds under the null hypothesis as well as under the alternative
showing that the bootstrap test is asymptotically equivalent to the asymptotic test. Theoretical justifications for the bootstrap procedure providing better small sample behavior are mainly available for simple test statistics such as the mean (see for example \cite{accbootstrap}). Therefore, simulation studies are usually performed in order to assess the size and power of a bootstrap procedure. In this work, we apply the bootstrap to the projections as resampling the functional observations would require the estimation of the covariance kernel for each bootstrap sample which is computationally infeasible. Whether this leads to theoretically justifiable bootstrap procedures remains to be seen in future work.\\
As discussed above, due to the non-pivotal limit distribution, resampling procedures are required to obtain critical values for our test. Aston and Kirch \cite{fda2} obtained reasonable results by applying multivariate block bootstrap procedures for the corresponding mean change procedure. We apply a circular block bootstrap to the $\mathfrak{d}:=d(d+1)/2-$dimensional sequence of the score products. 
In order to correct the data for a possible change we first estimate the change point in each component $i=1,\ldots, \mathfrak{d}$ as follows:
$$\hat{k}_i^*=\arg\max_{1\leq k\leq n}\left(\sum_{t=1}^k\hat{q}_{i}(t)-\frac kn\sum_{t=1}^n\hat{q}_{i}(t)\right),\quad\mbox{where } \hat{q}(t):=\vech[\hat{\eta}_t\hat{\eta}_t^T].$$
Thus, we can estimate the uncontaminated data by
\begin{equation}\label{res}
\tilde{q}_{i}(t)=\hat{q}_{i}(t)-\begin{cases}
\overline{\hat{q}_i^0},&1\leq t\leq \hat{k}_i^*,\\
\overline{\hat{q}_i^1},&t> \hat{k}_i^*,
\end{cases}
\end{equation}
where $\overline{\hat{q}_i^0}=\frac{1}{\hat{k}_i^*}\sum_{t=1}^{\hat{k}_i^*}\hat{q}_{i}(t)$ and $\overline{\hat{q}_i^1}=\frac{1}{n-\hat{k}_i^*}\sum_{t=\hat{k}_i^*+1}^n\hat{q}_{i}(t).$ 
We estimate the long-run variance of the original test statistic by
\begin{align}\label{lrest}
\hat{D}_n(i,i)=\frac{1}{n}\sum_{j=0}^{L-1}\left(\sum_{k=1}^K\tilde{q}_{i}(Kj+k)\right)^2,\quad \hat{D}_n(i,j)=0\quad\mbox{for } i\neq j,
\end{align}
where we use the same blocklength $K$ as in the following bootstrap procedure. We split the whole sequence of length n circularly into overlapping subsequences of length $K$ and repeat the following steps B times to obtain the bootstrap statistics $\tilde{\Omega}_n^{*(b)},b=1,\ldots,B$:
\begin{itemize}
\item[(1)] Draw the starting points of the blocks as realizations of $$U(0),\ldots,U(L)\stackrel{i.i.d.}{\sim}U(\{0,\ldots,n-1\})\quad\mbox{ with }L:=\left\lfloor\frac{n}{K} \right\rfloor.$$
\item[(2)] Generate a bootstrap sample by
$$q_{i}^*(Kj+k):=\tilde{q}_{i}(U(j)+k),\quad j=0,\ldots,L,\quad k=1,\ldots,K,\quad i=1,\ldots, \mathfrak{d},$$
where $\tilde{q}_{i}(t)=\tilde{q}_{i}(t-n)$ if $t>n.$
\item[(3)] Calculate residuals $\tilde{q}^*_{i}(t)$ of the bootstrap sample of length n analogously to (\ref{res}).
\item[(4)] Calculate $D^*_n$ by
\begin{align}\label{blockest}
D^*_n(i,i)=\frac{1}{n}\sum_{j=0}^{L-1}\left(\sum_{k=1}^K\tilde{q}^*_{i}(Kj+k)\right)^2,\quad D^*_n(i,j)=0\quad\mbox{for } i\neq j.
\end{align}
\item[(5)] Calculate the bootstrap statistic by
$$\tilde{\Omega}_n^*=\frac 1n\sum_{k=1}^n S^{*T}_kD^{*^{-1}}_nS^*_k,
$$
with $S^*_k=(S^*_k(1),\ldots, S^*_k(\mathfrak{d}))^T,$ $S^*_k(i)=\frac{1}{\sqrt{n}}\left(\sum_{t=1}^k(q^*_{i}(t)-\overline{q}_{n,i}^*)\right),$ $\overline{q}_{n,i}^*=\frac 1n\sum_{t=1}^nq^*_{i}(t).$ 
\end{itemize}
We obtain the critical values as the upper $\alpha$-quantiles of the B realizations $\tilde{\Omega}_n^{*(b)},b=1,\ldots,B$. The validity of the corresponding multivariate block bootstrap has been shown in \cite{Weber} taking possible changes into account. In the functional setting this should carry over as long as the eigenvalues are well separated but a detailed theoretic analysis is beyond the scope of this paper.
\subsection{Some alternative test statistics}\label{alttest}
The main drawback of change point procedures based on dimension reduction techniques is their inability to detect changes which are orthogonal to the projection space as given by condition (\ref{changedet}) for the covariance change. Furthermore, the asymptotic distributions do not yield reasonable small sample approximations if the dimension of the projection space is chosen too large. This is particularly problematic when testing for a covariance change as the procedure is based on the $d(d+1)/2$-dimensional product vector when projecting on a d-dimensional subspace. Even if we only use the first two leading eigenfunctions of each direction in the separable dimension reduction and thus risk missing possible changes which do not occur in this very limited number of eigendirections we project on a 8-dimensional subspace and obtain 36-dimensional product vectors. Taking 3 eigenfunctions in each direction results in a 378-dimensional product vector which is considerably larger than the sample size and thus problematic for the multivariate procedure. This motivates us to consider fully functional test statistics.
Recall that after reducing the dimension, the test statistic as given in (\ref{statadapt}) is based on
\begin{align}\label{tkd}
T_k=S_k^TD_n^{-1}S_k=\frac 1n\sum_{l_1=1}^{d}\sum_{l_2=l_1}^{d}\frac{1}{\hat{\gamma}^2_{l_1,l_2}}\left(\sum_{t=1}^k\left(\hat{\eta}_{t,l_1}\hat{\eta}_{t,l_2}-\overline{\hat{\eta}_{l_1}\hat{\eta}_{l_2}}\right)\right)^2
\end{align}
with $S_k$ as given in (\ref{sk1}), $\overline{\hat{\eta}_{l_1}\hat{\eta}_{l_2}}=\frac 1n\sum_{t=1}^n\hat{\eta}_{t,l_1}\hat{\eta}_{t,l_2}$ and $\hat{\gamma}^2_{l_1,l_2}$ is an estimator for $\gamma^2_{l_1,l_2}=\sum_{t\in\mathbb{Z}}\Cov\left(\hat{\eta}_{0,l_1}\hat{\eta}_{0,l_2},\hat{\eta}_{t,l_1}\hat{\eta}_{t,l_2}\right)$. The weight $\frac{1}{\hat{\gamma}^2_{l_1,l_2}}$ corrects for different variances in the time series of the score products making smaller changes in components with smaller variances better visible for the test statistic. This approach is related to the likelihood ratio statistic in the multivariate case. However, the price to pay is that changes - even big ones - in score components other than the first d will not be detected at all. This seems quite unnatural. Therefore, we consider alternative test statistics related to the procedures proposed in \cite{Wendler} and \cite{AueFull} for the mean change problem which take the full functional structure into account. An obvious and well defined alternative to reducing the dimension is
\begin{align}\label{tkf}
T_k^F=\frac 1n \sum_{l_1=1}^{\infty}\sum_{l_2=l_1}^{\infty}\left(\sum_{t=1}^k(\eta_{t,l_1}\eta_{t,l_2}-\overline{\eta_{l_1}\eta_{l_2}})\right)^2
\end{align}
which takes all scores of the basis expansion into account but without correcting for different variances as the multivariate test statistic does. Due to the squared summability of the eigenvalues, this infinite sum is well defined. In order to keep the advantage of $T_k$ in terms of the weights improving the visibility of changes in components with smaller variances while not risking to miss a change due to dimension reduction we suggest the following weighting
\begin{align}\label{tkw}
T_k^W=\frac 1n \sum_{l_1=1}^{\infty}\sum_{l_2=l_1}^{\infty}\frac{1}{s_{1,1}^2+\hat{\gamma}^2_{l_1,l_2}}\left(\sum_{t=1}^k(\hat{\eta}_{t,l_1}\hat{\eta}_{t,l_2}-\overline{\hat{\eta}_{l_1}\hat{\eta}_{l_2}})\right)^2,
\end{align}
where $s_{1,1}^2$ is the estimated variance of the first squared score component. This additive constant is needed for bounding the denominator of the weights away from zero and is chosen such that the test statistic is scale invariant. By (\ref{normalsig}) for independent Gaussian scores the variance of the first squared score component is given by $2\lambda_1^2$ which is the largest element in the long-run covariance matrix.\\
We calculate the critical values with an analogous the bootstrap procedure as described in Section \ref{resampling}. For the weighted functional procedure the long-run variances are estimated for each bootstrap sample with the block estimator as in step (4) whereas we keep the variance of the first squared score component fixed. Analogously to the multivariate procedure we also use the block estimator \eqref{lrest} for estimating the long-run covariance for the test statistics.
\begin{remark}\label{rem.Tf}
$T_k^F$ is related to the statistic $\|S_k^F\|^2$, where
\begin{align*}
S_k^F(u,s)=\frac{1}{\sqrt{n}}\sum_{t=1}^k\left(X_t(u)X_t(s)-\overline{X(u)X(s)}\right),
\end{align*}
with $\overline{X(u)X(s)}=\frac 1n\sum_{t=1}^nX_t(u)X_t(s)$. More precisely, some calculations show (see (\ref{l2s}) in the supplementary material) that
\begin{align*}
\|S_k^F\|^2=\frac 1n \sum_{l_1,l_2=1}^{\infty}\left(\sum_{t=1}^k(\eta_{t,l_1}\eta_{t,l_2}-\overline{\eta_{l_1}\eta_{l_2}})\right)^2.
\end{align*}
In contrast to $T_k^F$, this statistic contains all combinations of $l_1\neq l_2$ twice such that the cross-covariances have double weights compared to the variances. This is an artefact when dealing with a bivariate symmetric function which does not occur in the mean change problem. In accordance with $T_k$ we construct the functional statistic $T_k^F$ such that each combination is only contained once.
\end{remark}
\section{Simulation study}\label{sim}
In the following simulation study we assess the empirical size and power of the proposed procedures. As there are no mathematical justifications for the bootstrap procedures for the functional test statistics available yet, the simulation study is of particular interest to evaluate their performance. Independent innovations $e_t(s)=\sum_{l=1}^D\eta_{t,l}v_l(s),t=1,\ldots,n,$ of length $n=200$ are generated using a Fourier basis with $D=55$ basis functions $\{v_1,\ldots,v_{55}\}$ on $[0,1]$, where $v_1(s)\equiv 1$ followed by pairs of $\sin(i\cdot s)$ and $\cos(i\cdot s)$ for $i=2,\ldots,27$. The scores $\{\eta_{t,l}:l=1,\ldots,55\}$ are independent and normally distributed with standard deviations $\{\sigma_{l}:l=1,\ldots,55\}$. Following the simulation study in \cite{AueFull} we consider the following settings:
\begin{itemize}
\item[]Setting 1 (small number of nonzero eigenvalues): $\sigma_l=1$ for $l=1,\ldots,8$ and $\sigma_l=0$ for $l=9,\ldots, 55$.
\item[]Setting 2 (fast decay of eigenvalues using): $\sigma_l=3^{-l}, l=1,\ldots,55$.
\item[]Setting 3 (slow decay of eigenvalues using): $\sigma=l^{-1}, l=1,\ldots,55$.
\end{itemize}
\begin{figure}[H]
\begin{subfigure}[t]{0.32\textwidth}
\centering
	\includegraphics[width=\textwidth]{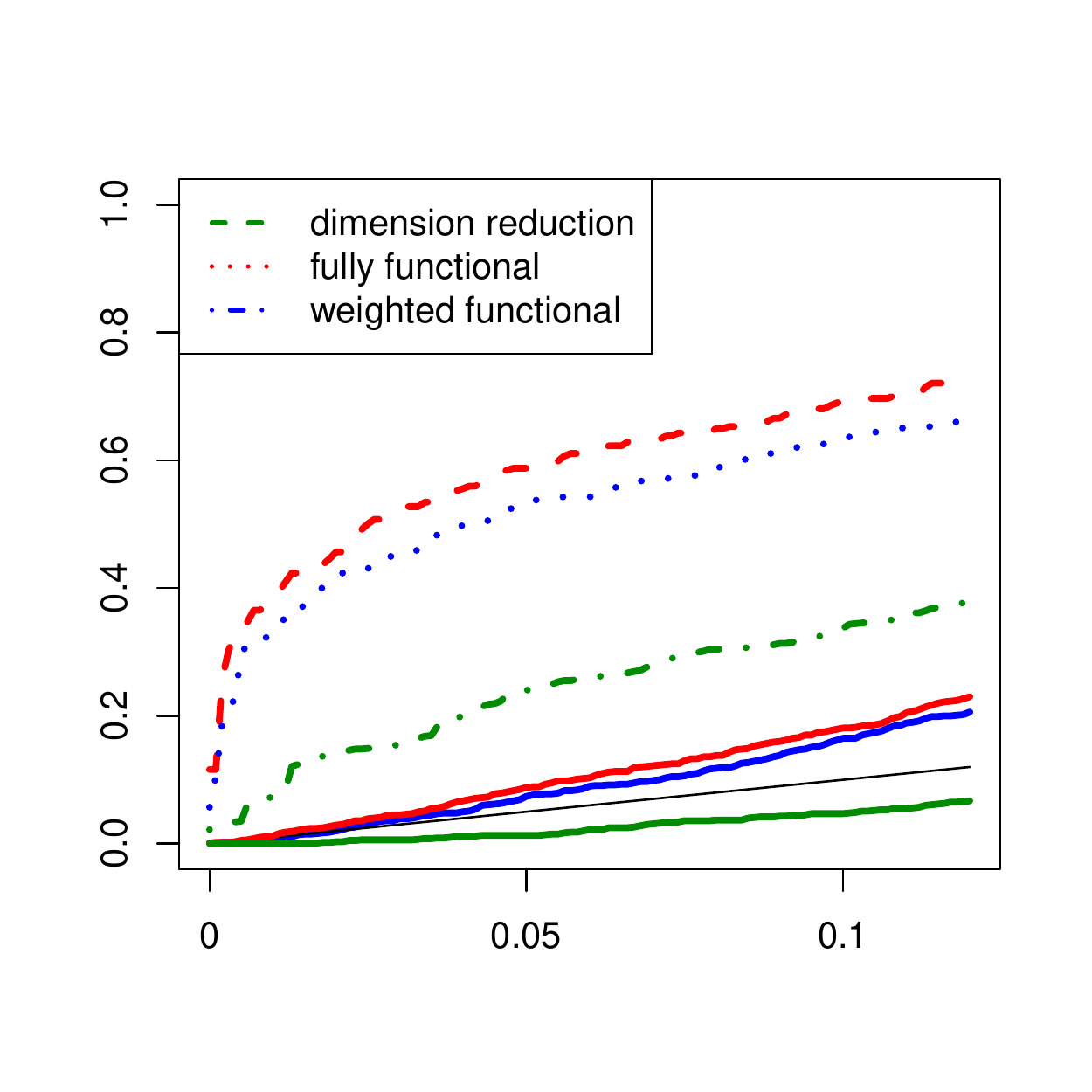}
	\caption{Setting 1,$\sigma_{\epsilon}=1,m=2,25,50$}
	\end{subfigure}
\hfill
\begin{subfigure}[t]{0.32\textwidth}
	\centering
	\includegraphics[width=\textwidth]{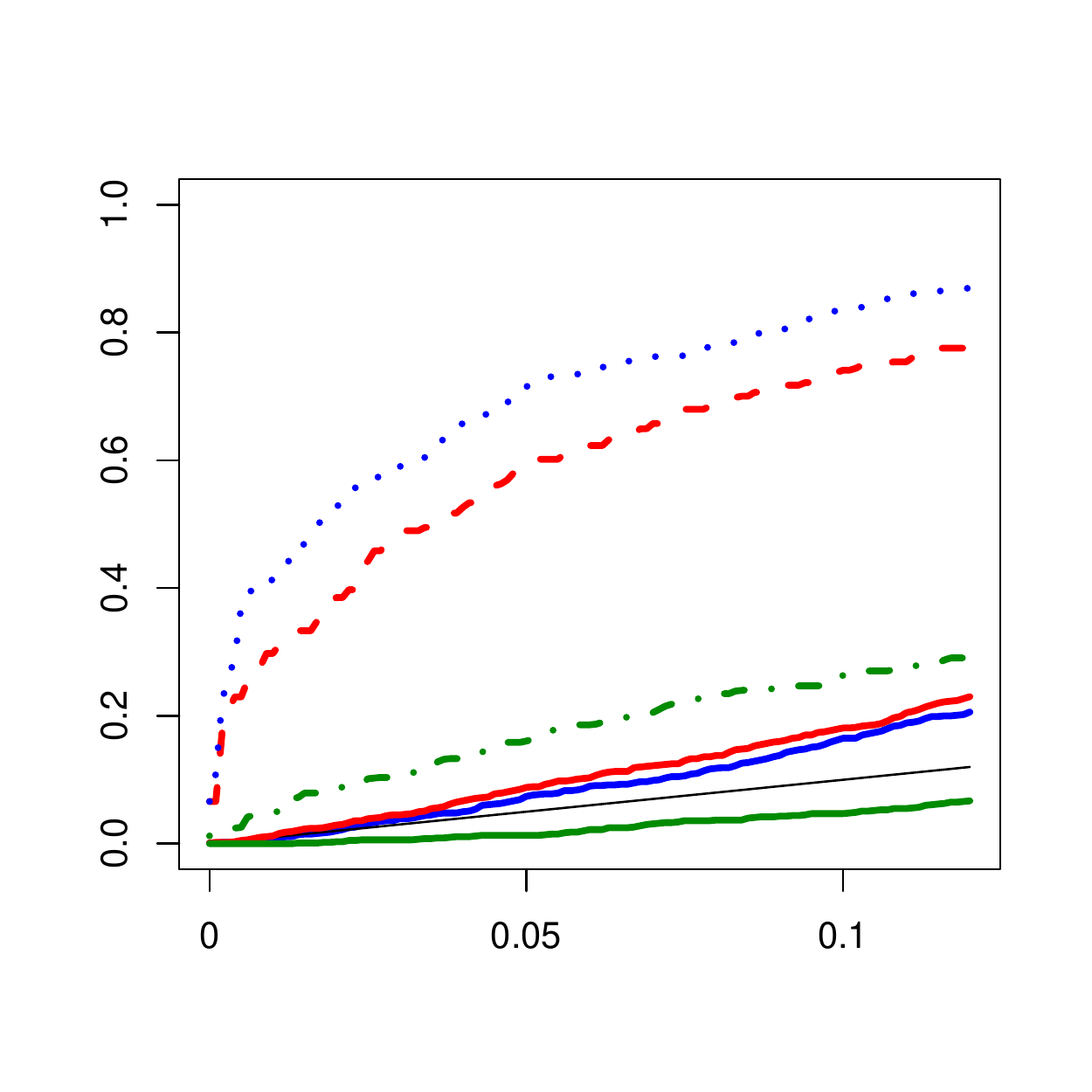}

\end{subfigure}
\hfill
\begin{subfigure}[t]{0.32\textwidth}
	\centering
	\includegraphics[width=\textwidth]{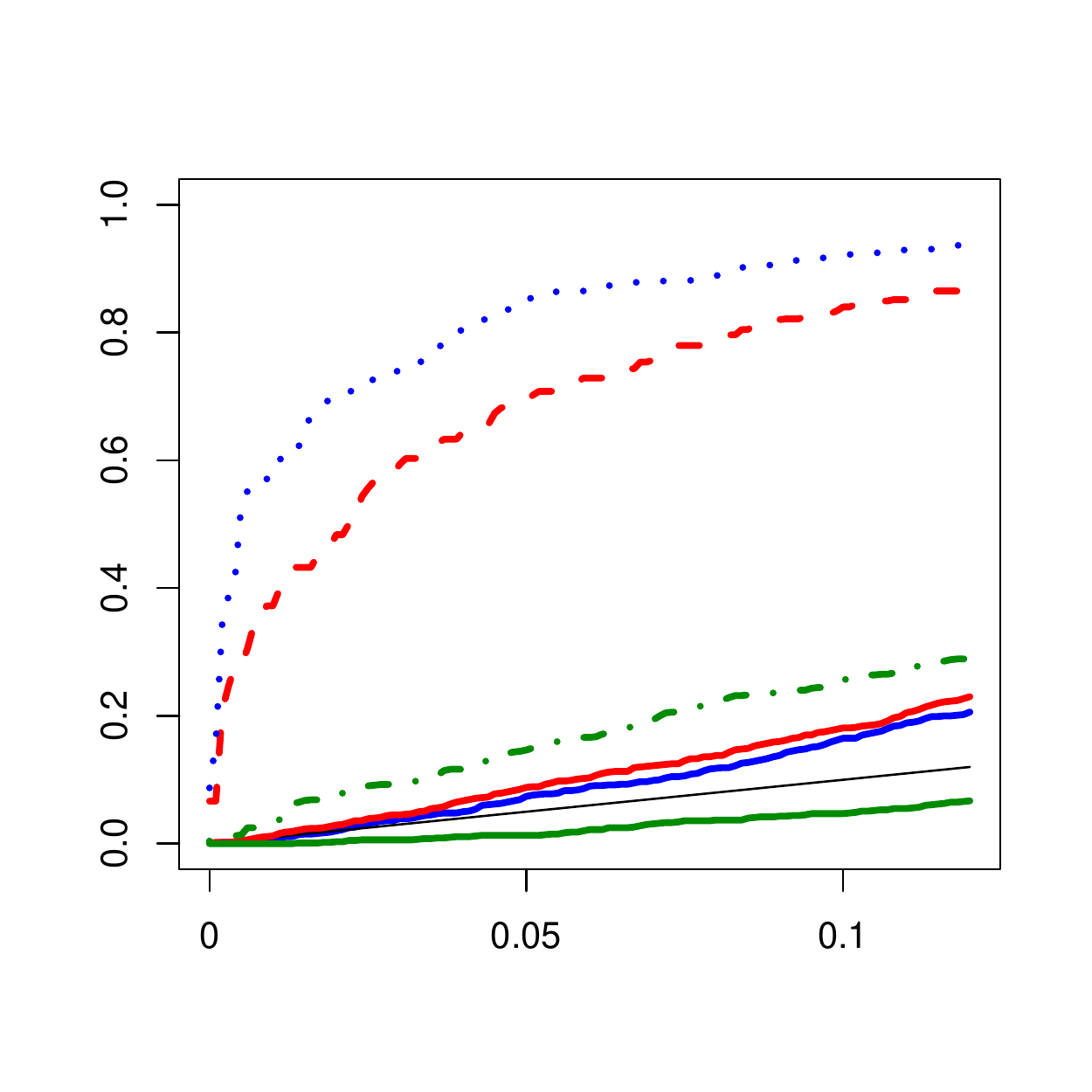}

\end{subfigure}
\begin{subfigure}[t]{0.32\textwidth}
\centering
	\includegraphics[width=\textwidth]{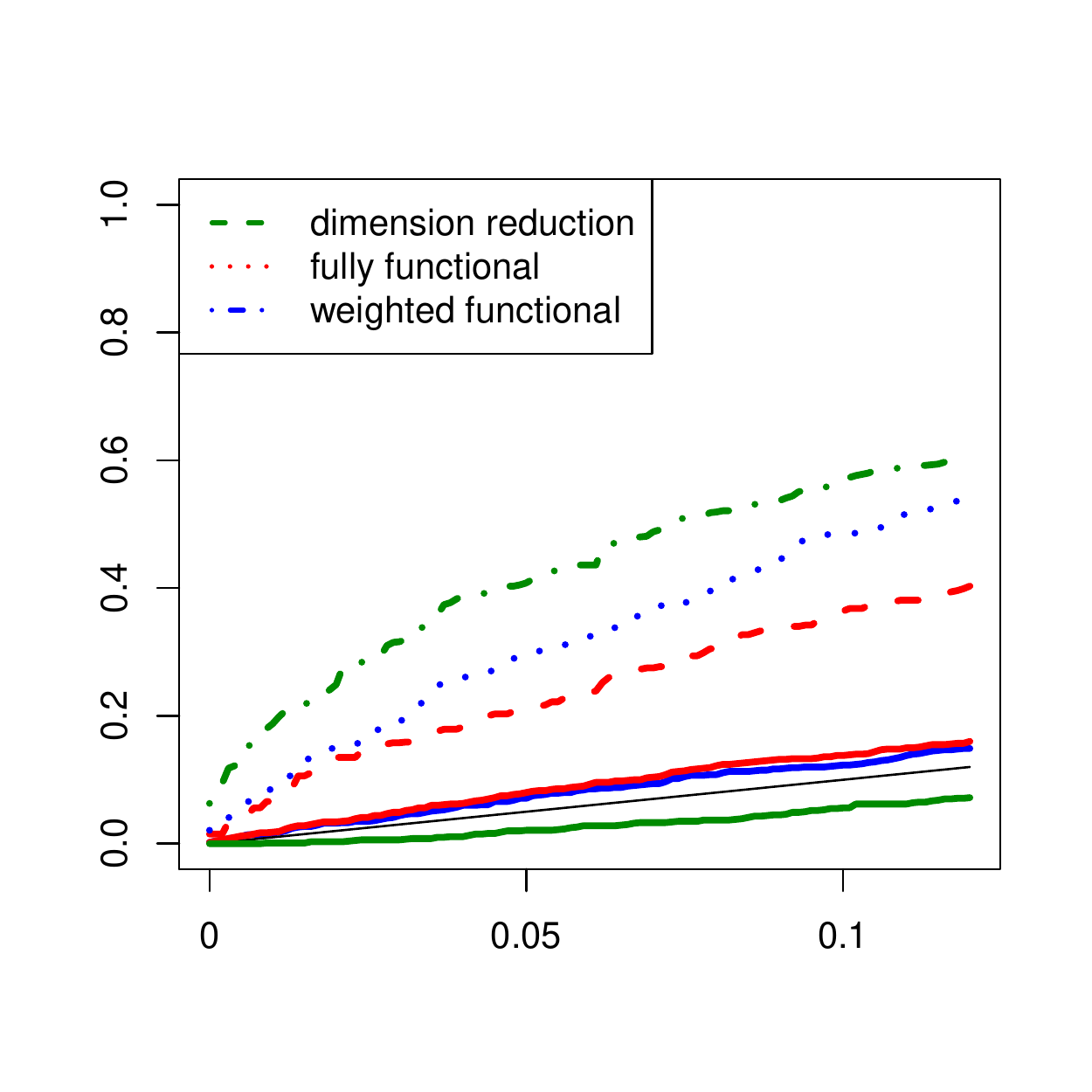}
	\caption{Setting 2, $\sigma_{\epsilon}=0.2,m=2,25,50$}
	\end{subfigure}
\hfill
\begin{subfigure}[t]{0.32\textwidth}
	\centering
	\includegraphics[width=\textwidth]{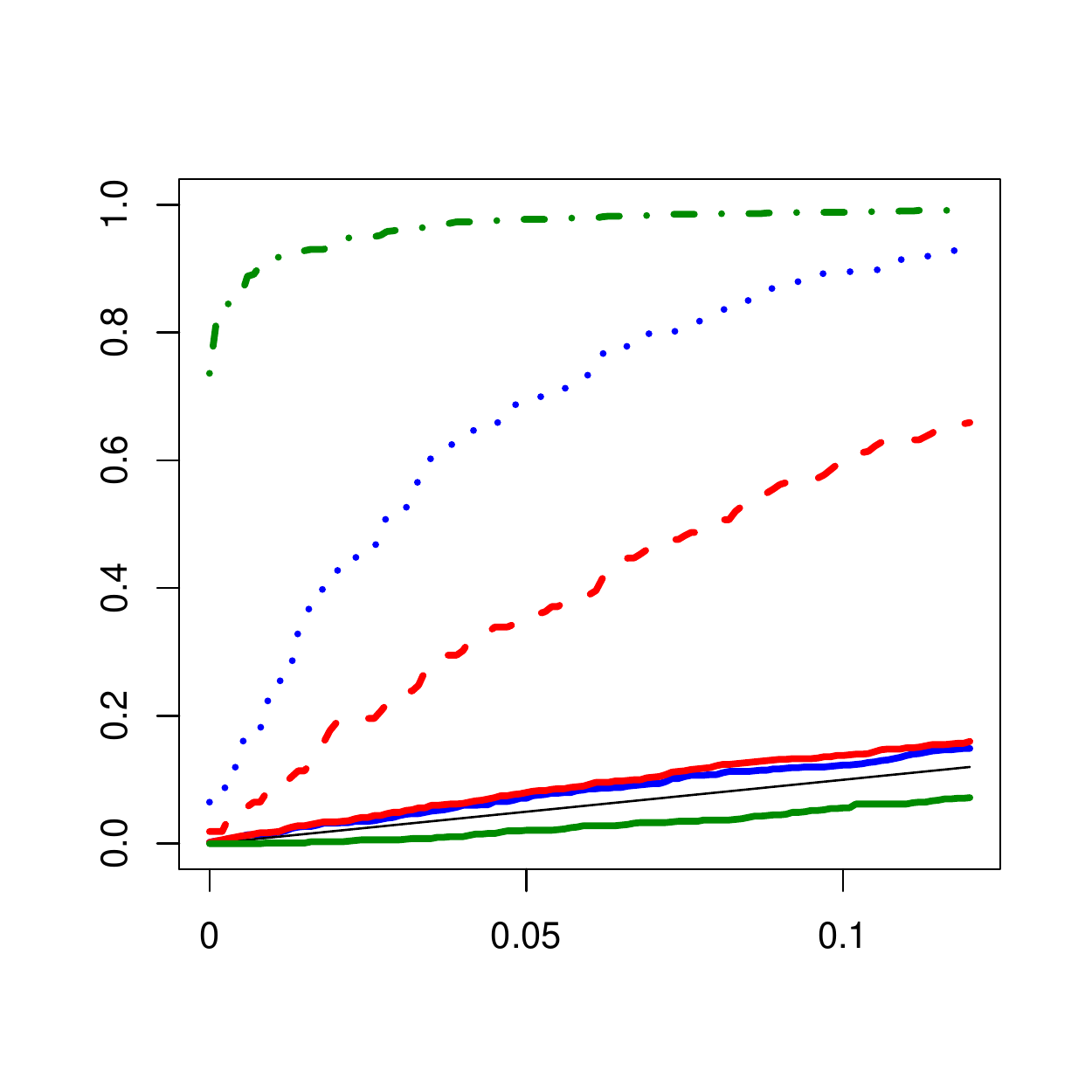}

\end{subfigure}
\hfill
\begin{subfigure}[t]{0.32\textwidth}
	\centering
	\includegraphics[width=\textwidth]{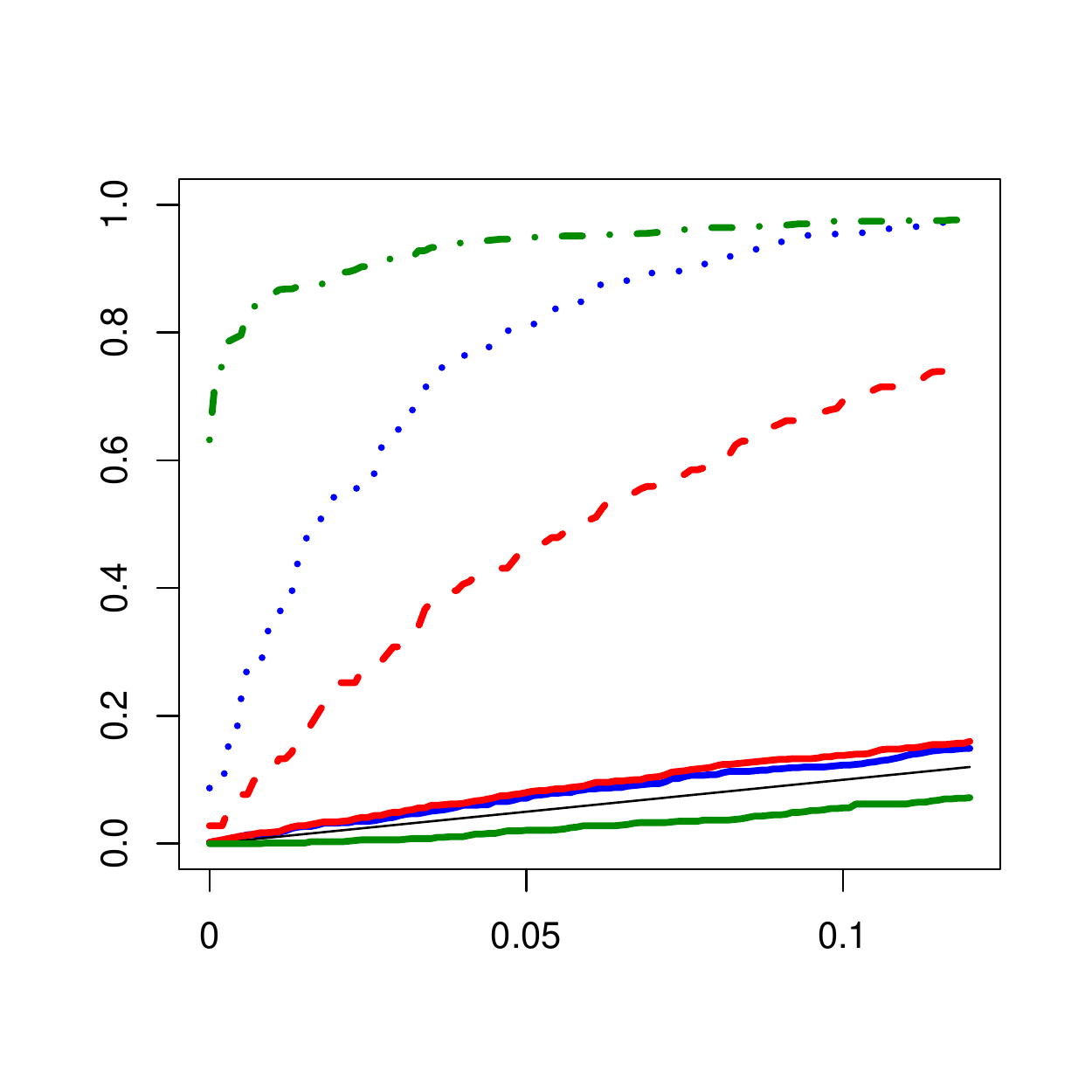}

\end{subfigure}
\begin{subfigure}[t]{0.32\textwidth}
\centering
	\includegraphics[width=\textwidth]{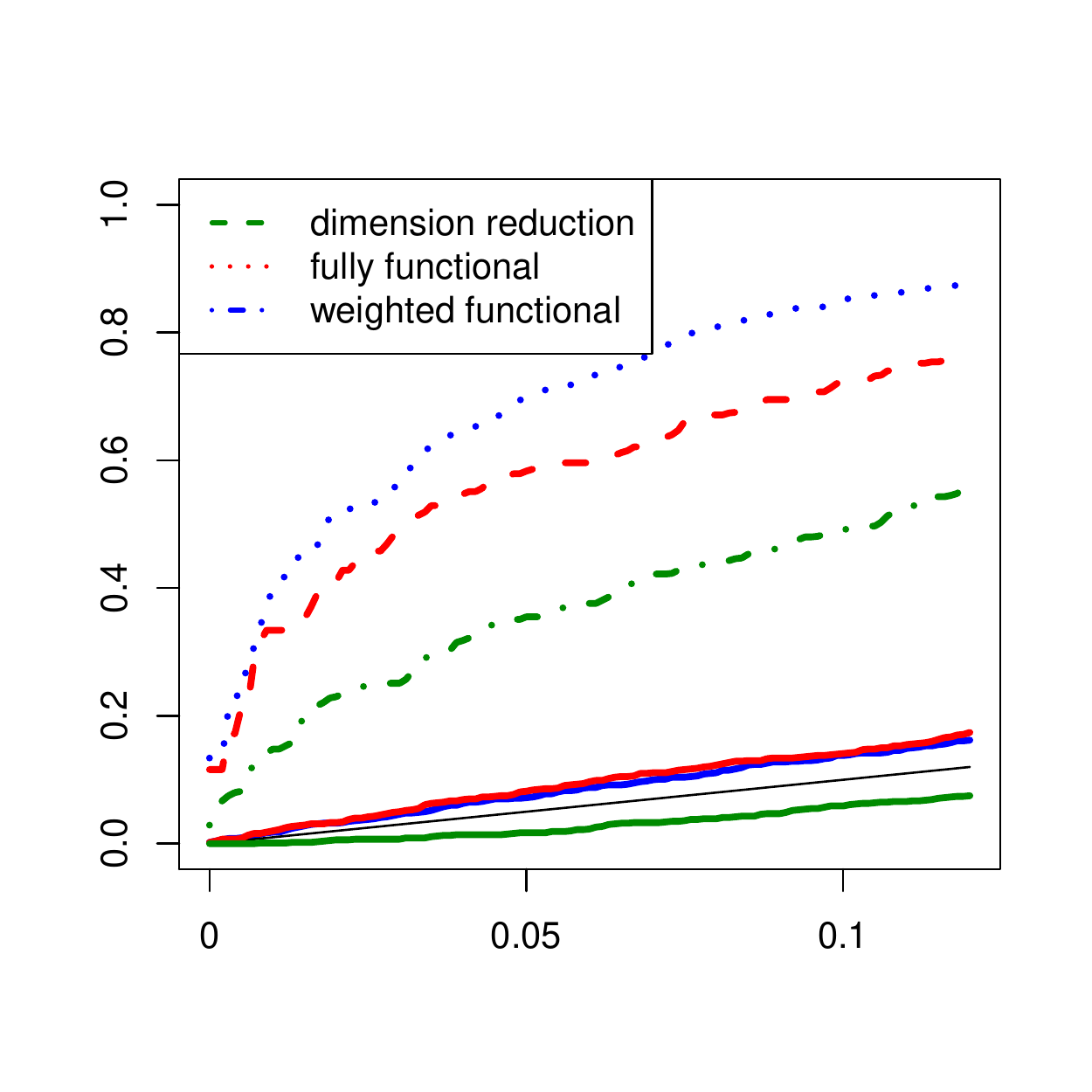}
	\caption{Setting 3, $\sigma_{\epsilon}=0.8,m=2,25,50$}
	\end{subfigure}
\hfill
\begin{subfigure}[t]{0.32\textwidth}
	\centering
	\includegraphics[width=\textwidth]{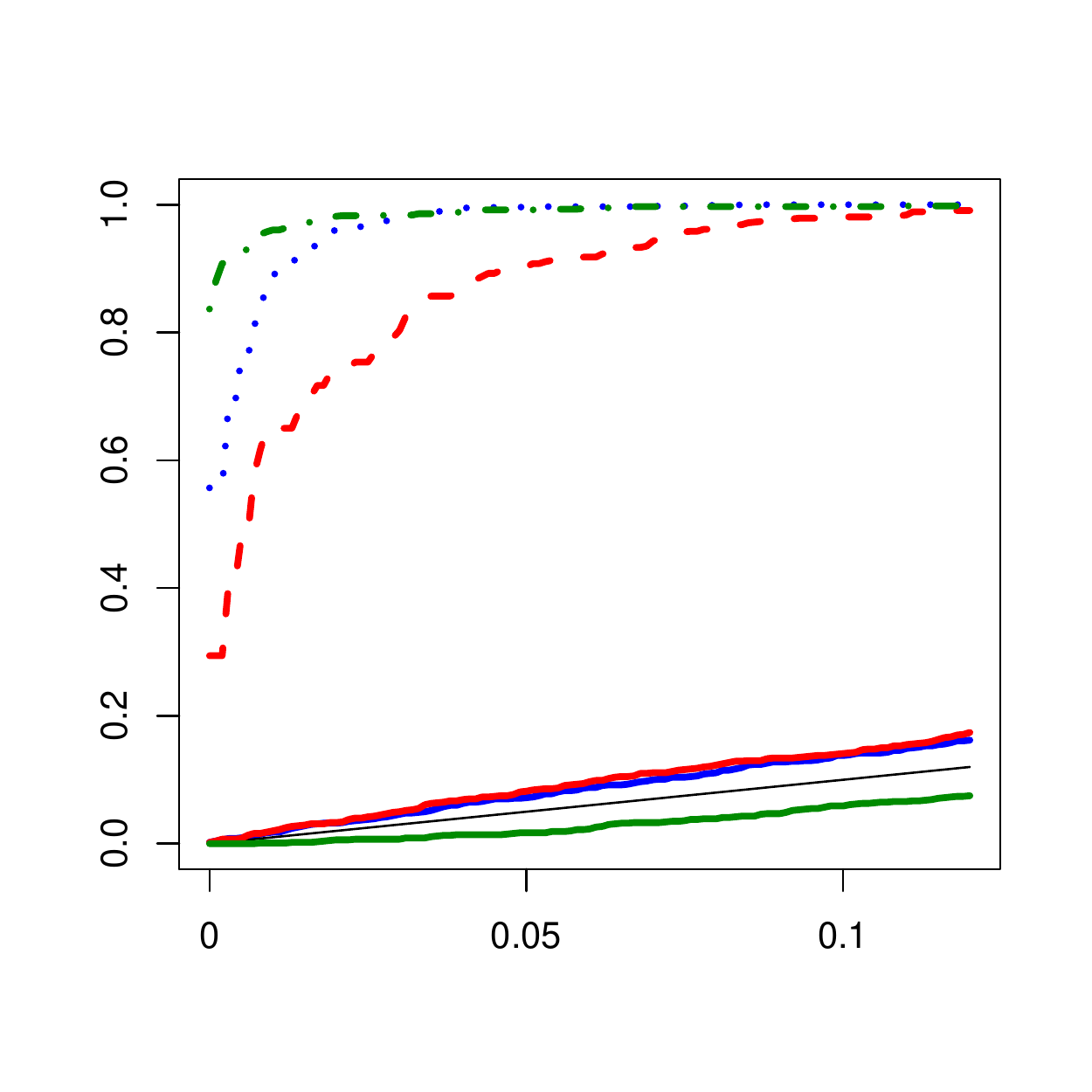}

\end{subfigure}
\hfill
\begin{subfigure}[t]{0.32\textwidth}
	\centering
	\includegraphics[width=\textwidth]{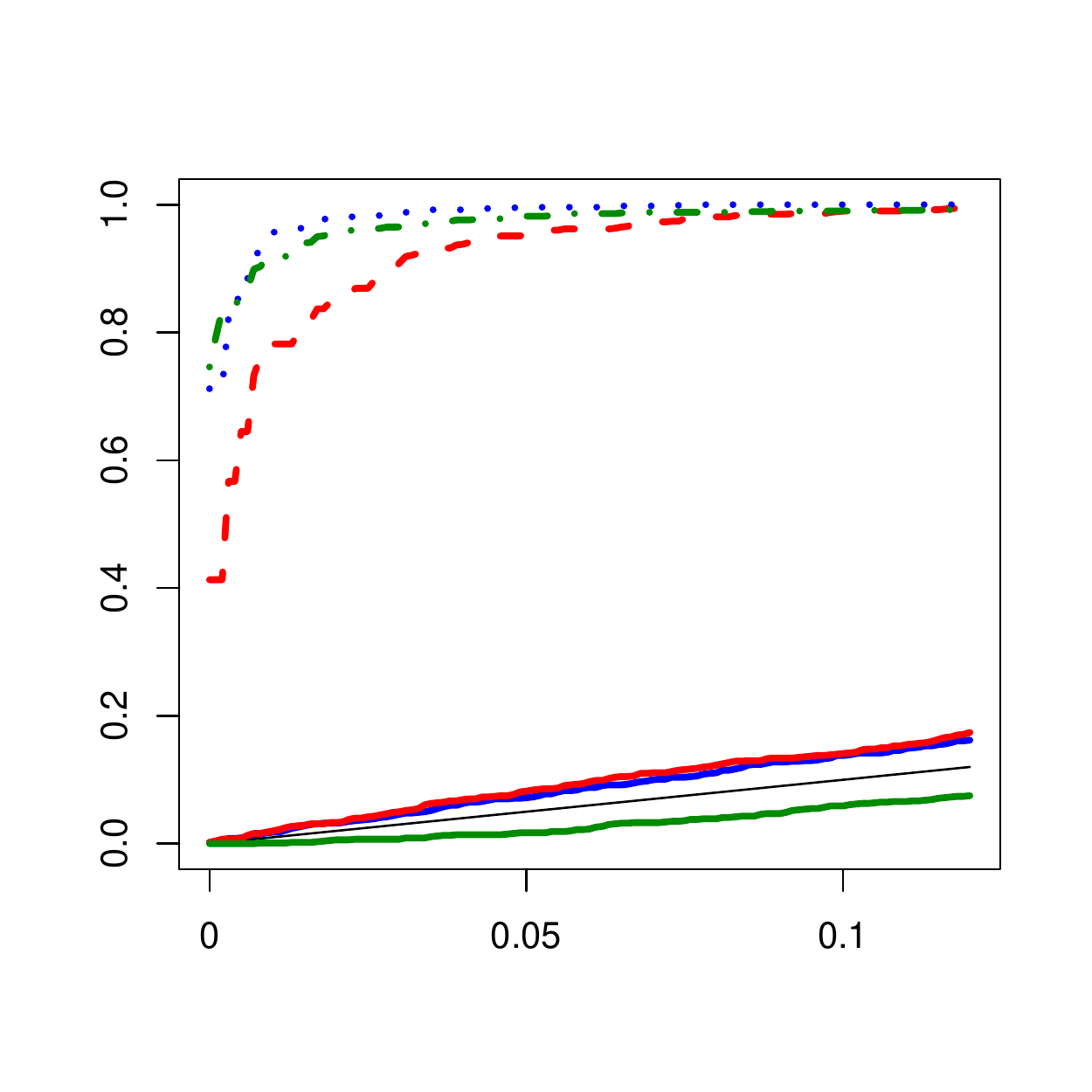}
\end{subfigure}
\caption{Empirical size (solid lines) and size corrected power (dashed lines) of the proposed procedures with $K=6$ for functional autoregressive time series using the multivariate procedure after dimension reduction based on \eqref{tkd} (green), the fully functional procedure based on \eqref{tkf} (red) and the weighted functional procedure based on \eqref{tkw} (blue).}
\label{fig.sim.far}
\end{figure}
Functional autoregressive time series $X_t=\Psi(X_{t-i})+e_t$ are simulated where the linear operator $\Psi$ can be represented as a $D\times D$-matrix that is applied to the coefficients of the basis representation via $\{v_1,\ldots,v_{55}\}$ (for further details see \cite{Auesim}). In this simulation study we use the operator with $0.4$ on the diagonal and $0.1$ on the superdiagonal and the subdiagonal which has infinity norm $0.6$ such that the resulting functional autoregressive time series is stationary. A covariance change at the time point $0.5n$  is inserted in the first $m$ leading eigendirections for $m=2,25,50$ by adding a common additive noise term $\epsilon_{t,l}=\epsilon_t$ with variance $\sigma_{l_1,l_2}=\frac{\sigma_{\epsilon}^2}{m}$ according to Example \ref{ex2}. The variance of the noise term is chosen such that $\int\int\delta^2(u,s)du\,ds=1$ for all m. In view of the application to fMRI data in Section \ref{appl} the multivariate procedure is applied to the projections on the subspace spanned by the first 8 eigendirections of the empirical covariance function. The plots in Figure \ref{fig.sim} show the empirical size and the size corrected power for the different procedures obtained based on N=1000 repetitions with B=1000 bootstrap iterations each.\\

The multivariate procedure is very conservative in all settings whereas the size of the functional procedures is mostly larger but closer to the nominal level except for setting 1. However, it should be mentioned that for independent data (see Figure \ref{fig.sim} in the supplementary material) all procedures keep the level very well in all settings. 
As expected by construction, the procedure based on PCA fails to detect the change in setting 1 for increasing m as most of the change is orthogonal to the first 8 eigendirections which are still dominating the contaminated covariance kernel. The advantage of the procedures which take the full functional structure into account is clearly visible here. The opposite power behavior can be observed for the fast decay of eigenvalues in setting 2, where the procedure based on PCA is superior to the functional procedures. In particular the unweighted functional procedure has problems to detect the change in this setting. In setting 3, the functional procedures have good power for all choices of m. In applications where one aims to explain a large amount of the variability of the data via PCA, a slow decay of eigenvalues as in setting 3 usually leads to a bad performance. However, this is not true when PCA is applied for change point detection if the change leads to an increased variability in the affected directions which is true for the alternative in this simulation study. Hence, directions which are affected by the change but orthogonal to the uncontaminated subspace are more likely to be chosen by PCA if the eigenvalues are flat. This effect can be observed when comparing the power of the multivariate procedure for $m=50$ in setting 2 and 3. For $m=2$ the power of the multivariate procedure is slightly better in setting 2 than in setting 3 as for the fast decay of eigenvalues the change occurs in those eigendirections which already clearly dominate in the uncontaminated subspace. Across all situations considered in this simulation study except for setting 1 with $m=2$ the weighted functional procedure outperforms the unweighted functional procedure. Hence, the weighted functional procedure behaves not only as a compromise between the other two procedures but even more as a promising improvement of the unweighted functional approach.
\subsection*{Discussion}
The above simulation study reveals that the functional test statistics with critical values obtained by the block bootstrap described in Section \ref{resampling} can be liberal for dependent data. This is mostly a small sample effect which did not occur in simulations of longer time series (T=500). However, even for the sample size in the present simulations, the size is reasonable up to a nominal level of $5\%$ and the procedures seem to be suitable for the purpose of the application in this paper. We do not expect them to cause too many false rejections and we are in particular interested in a good power behavior in order to avoid nonstationarities to contaminate subsequent analyzes. For future work, it would be of great interest to investigate the mathematical validity of this bootstrap approach as well as develop procedures that improve the behavior of the functional procedures for dependent data and can also deal with stonger dependency structures.

\section{Application to resting state fMRI data}\label{appl}
In \cite{fda2} a subset of 198 scans from the \textit{1000 Connectome Resting State Data}\footnote{The data is publicly available from the International Neuroimaging Data-Sharing Initiative (INDI) at \url{http://fcon_1000.projects.nitrc.org}.} which have all been recorded at the same location (Beijing, China) are tested for an epidemic mean change. We test for deviations from covariance stationarity in those 118 data sets among these where no epidemic mean change was detected at a level of $5\%$ in the previously mentioned work. Each scan consists of a three-dimensional image of size $64\times 64\times 33$ ($\sim 10^5$ voxels) recorded every 2 seconds at 225 time points. As usual in fMRI data analysis, each data set is preprocessed by voxelwise removing a polynomial trend of order 3 to correct for technical effects as for example scanner drift. We apply the separable covariance estimation and for the multivariate procedure we reduce the dimension by projecting on the 8-dimensional subspace obtained by taking the first two eigenfunctions in each direction.
\subsection{Implementation of the functional procedures}
In practice, the sums in (\ref{tkf}) and (\ref{tkw}) are finite as we cut after the number N of strictly positive eigenvalues obtained by principal component analysis. In the above simulation study we obtained $N\approx 100$ but for the fMRI data sets the separable covariance estimation yields $N\approx 10^5$. For the functional test statistics all combinations of the score components have to be taken into account. As the number of those combinations is of order $10^{10}$ this is computationally infeasible, in particular with regard to the fact that the test statistic also has to be calculated for every single bootstrap sample. However, as the variances of the score products rapidly decrease, most of the score products only have a negligible influence on the value of the test statistic in comparison to those with large variances and can thus be omitted. 
\begin{figure}[H]
\begin{minipage}[t]{0.48\textwidth}
	\centering
	\includegraphics[width=\textwidth]{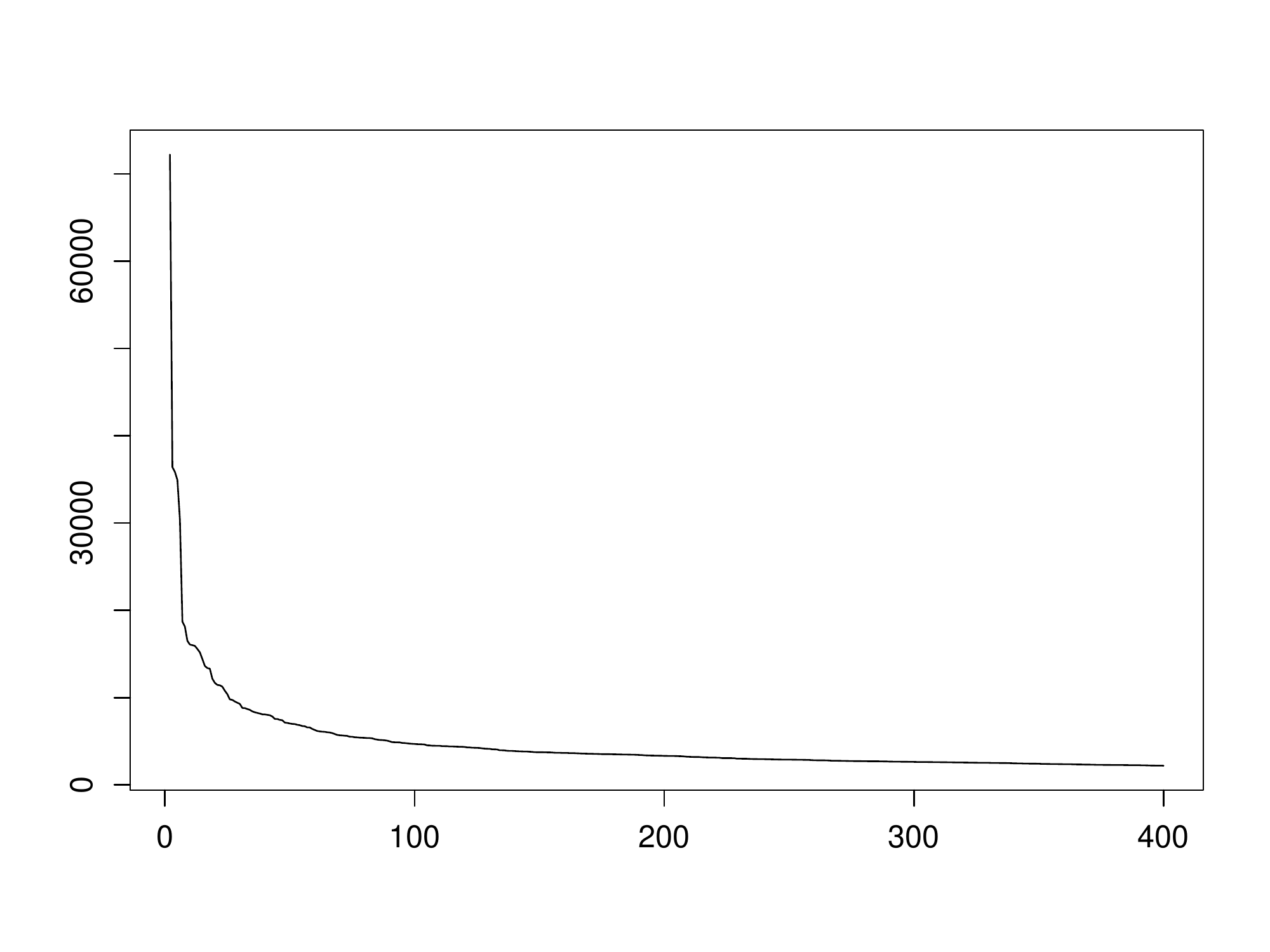}
	\caption{sub06880: 2nd to 200 largest variance of score products in decreasing order.}
	\label{varprod}
\end{minipage}
\hfill
\begin{minipage}[t]{0.48\textwidth}
	\centering
	\includegraphics[width=\textwidth]{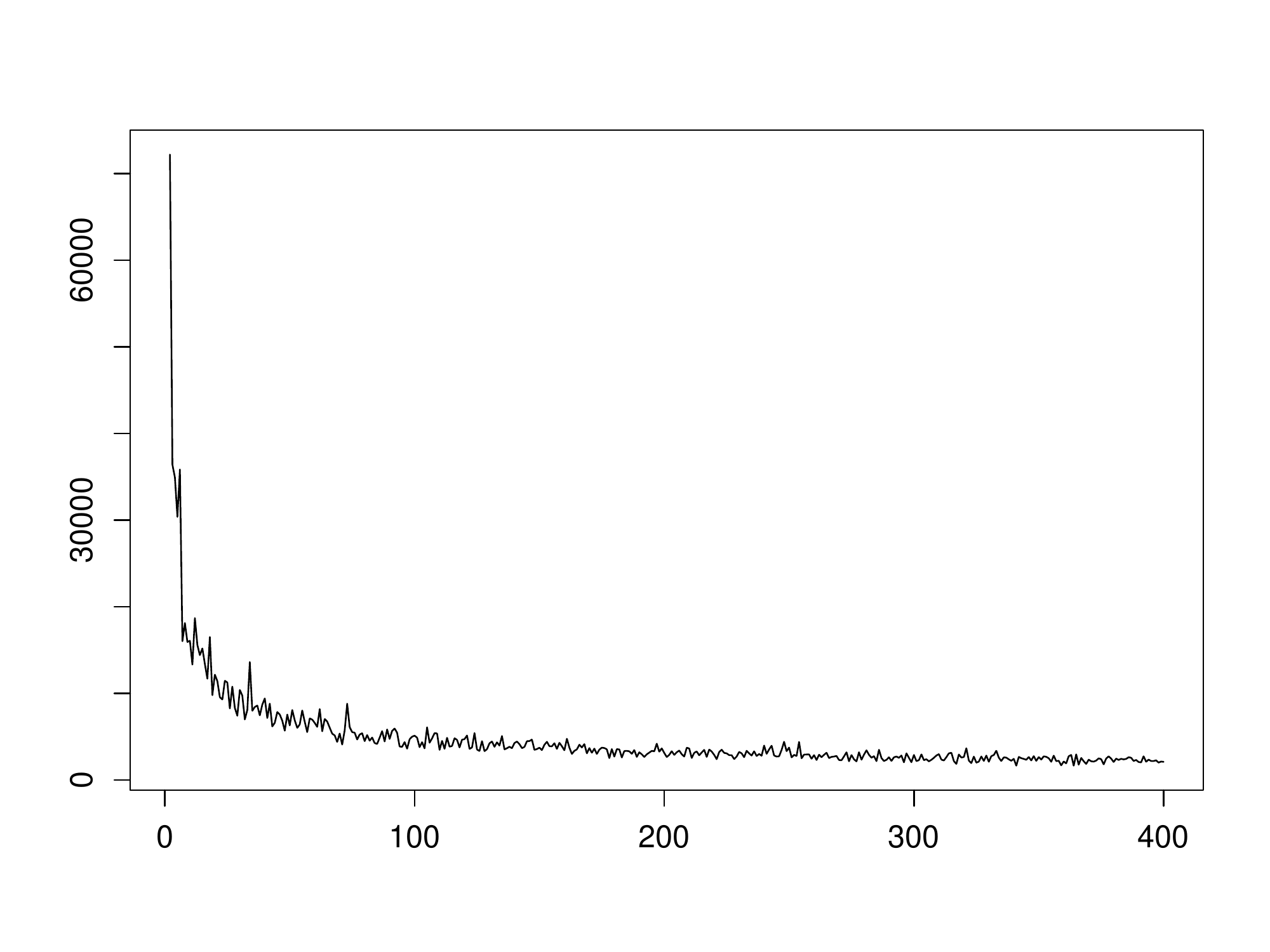}
	\caption{sub06880: 2nd to 200 largest variance of score products ordered according to their approximations given by the products of the variances of the single components.}
	\label{varest}
\end{minipage}
\end{figure}
Figure \ref{varprod} shows the 200 largest variances of the score products after correcting for a possible change in decreasing order exemplarily for one subject. The variance of the first squared score component is approximately 10 times larger than the second one and is thus omitted in this plot for a better visibility. It can clearly be seen that the variances strongly decrease and quickly level off at a magnitude which is only a small fraction of the larger variances. We make use of this observation to solve the computational problem discussed above where the main idea is to only consider those score products that have a sufficiently large variance compared to the variance of the first squared score component. However, estimating this ratio by calculating the empirical variance of the residuals for each of the $10^{10}$ combinations is still very time consuming. Therefore, we use a preselection step where we predict which combinations could possibly exceed a certain threshold based on the variances of the single score components. More precisely, we proceed as follows:
\begin{itemize}
\item[(1)] For each $l_1,l_2 =1,\ldots,N$ calculate
$$\hat{r}_{l_1,l_2}=\begin{cases}
\frac{s_{l_1}s_{l_2}}{2s_{1}^2},&l_1\neq l_2,\\
\frac{s_{l_1}^2}{s_1^2},&l_1= l_2
\end{cases}\quad\mbox{with}\quad s_{l}^2=\frac{1}{n-1}\sum_{t=1}^n\tilde{\eta}_{l}(t)^2,$$
where $\tilde{\eta}_{l}(t)$ is the estimated residual of $\hat{\eta}_{t,l}$ obtained as in (\ref{res}). Determine for $\epsilon_1=0.0005$ $$P:=\{(l_1,l_2):l_1,l_2=1,\ldots,N,\hat{r}_{l_1,l_2}\geq\epsilon_1\}.$$ This estimation of the ratio is based on the Gaussian approximation as given in (\ref{normalsig}). While this is only correct in the Gaussian case and if the separability assumption is correct, according to some preliminary analysis (see Figure \ref{varest}) it at least approximates the order of magnitude in the misspecified case. Figure \ref{varest} shows the variances of the score products ordered according to their approximations given by the products of the variances of the single components.
\end{itemize}
\begin{itemize}
\item[(2)] Perform the following steps for each $(l_1,l_2)\in P$:
\begin{itemize} 
\item[(2.1)] Estimate the ratio nonparametrically (without relying on Gaussanity or the separability assumption) by
$$r_{l_1,l_2}=\frac{s_{l_1,l_2}^2}{s_{1,1}^2}\quad\mbox{with}\quad s_{l_1,l_2}^2=\frac{1}{n-1}\sum_{t=1}^n\left(\widetilde{\eta_{l_1}\eta_{l_2}}(t)\right)^2,$$
where $\widetilde{\eta_{l_1}\eta_{l_2}}(t)$ is the estimated residual of the product $\hat{\eta}_{t,l_1}\hat{\eta}_{t,l_2}$ obtained analogously to (\ref{res}).
\item[(2.2)] If $r_{l_1,l_2}\geq\epsilon_2=0.0025$ continue with step (2.3), else skip this combination and continue with step (2.1) for the next combination.
\item[(2.3)] Update
\begin{align*}
T_k^W&=T_k^W+\frac{1}{s_{1,1}^2+\hat{\gamma}^2_{l_1,l_2}}\left(\sum_{t=1}^k(\hat{\eta}_{t,l_1}\hat{\eta}_{t,l_2}-\overline{\hat{\eta}_{l_1}\hat{\eta}_{l_2}})\right)^2,\quad k=1,\ldots,n\\
T_k^F&=T_k^F+\left(\sum_{t=1}^k(\hat{\eta}_{t,l_1}\hat{\eta}_{t,l_2}-\overline{\hat{\eta}_{l_1}\hat{\eta}_{l_2}})\right)^2,\quad k=1,\ldots,n.
\end{align*}
with $\hat{\gamma}^2_{l_1,l_2}=\frac{1}{n}\sum_{j=0}^{L-1}\left(\sum_{k=1}^K\widetilde{\eta_{l_1}\eta_{l_2}}(Kj+k)\right)^2$, where K is the block length of the respective bootstrap procedure and $L:=\left\lfloor\frac{n}{K} \right\rfloor$.
\end{itemize}
\item[(3)] Calculate the test statistics: $\Omega_n^W=\frac{1}{n}\sum_{k=1}^nT_k^W\quad\mbox{and}\quad\Omega_n^F=\frac{1}{n}\sum_{k=1}^nT_k^F.$
\end{itemize}
We additionally applied the procedure with $\epsilon_2=0.005$ to the resting state fMRI data and the results are similar to those obtained for $\epsilon_2=0.0025$. Hence, there is no need to further reduce the threshold as there is already no considerable loss of information when reducing it from $0.005$ to $0.0025$. In the preselection step we find those combinations for which $\hat{r}_{l_1,l_2}\geq \epsilon_1$ with a very conservative threshold $\epsilon_1=0.0005$. In the above example the predicted ratio $\hat{r}_{l_1,l_2}$ is at most $1.2$ times larger than the actual ratio such that  $\epsilon_1=0.0005$ is indeed very conservative. We calculate the critical values analogously to the bootstrap procedure described in Section \ref{resampling}. For the weighted procedure the long-run variances are estimated for each bootstrap sample with the block estimator as in step (4) whereas we keep the variance of the first squared score component fixed.
\subsection{Results}\label{results}
In this section, we refer to the $p$-values obtained for $\epsilon_1=0.0005, \epsilon_2=0.0025$ and a blocklength of $K=\sqrt[3]{225}\approx 6$. The results of the data analysis are illustrated exemplary by the score products of certain subjects as a change in the covariance structure is visible as a mean change in the products which is indicated by the black line in the plots. However, as the functional procedures are, on average, based on around 10000 score products, the plots are limited to the 64 most significant products in the sense of having the smallest $p$-values which are obtained by componentwise calculating the $p$-values for the weighted functional statistic based on the respective bootstrap components. The main findings of the data analysis can be summarized as follows, for further details see \ref{results.app} in the supplementary material:
\begin{itemize}
\item When testing for the AMOC alternative at a level of $5\%$, the null hypothesis of covariance stationarity is rejected for $43\%$ of the data sets by the multivariate procedure, for $39\%$ by the unweighted functional procedure and for $36\%$ by the weighted functional procedure. The functional procedures always lead to similar results whereas the multivariate procedure implies different test decisions in some cases. Those deviation occur in both directions and are explained in more detail in the supplementary material. As an example, in sub12220 a covariance change is detected by all considered procedures with $p$-values of at most $0.001$.  Figure \ref{amoc.example1} shows the 64 most significant components of the score products for the weighted functional procedure. The estimated global change point is $\hat{k}^*=57.$
\item There are some data sets with epidemic changes. For example, sub08816 is not significant when testing for the AMOC alternative with a $p$-value of $0.11$ for the multivariate procedure and at least $0.36$ for the functional procedures whereas the test for the epidemic alternative yields $p$-values which are smaller than $0.04$ for the functional procedures. Figure \ref{ep.example2} shows the 64 most significant components of the score products for the epidemic alternative. The respective plots for the AMOC alternative can be found in Figure \ref{ep.example1} in the supplementary material.
\item Some data sets contain outliers which cause the rejection of the null hypothesis. For example, testing for an epidemic change in sub08992 yields $p$-values smaller than $0.05$ for all considered procedures. Figure \ref{out.example2} reveals that the procedure picks the outlier as epidemic change in form of a very small interval. The mean of this interval is obviously much larger than the mean of the remaining observations and additionally always at the same position determined by the outlier such that the test for an epidemic change is significant. Although, in this case, the rejection of the null hypothesis is not due to an actual change in the covariance structure, an outlier constitutes a deviation from stationarity which contaminates the subsequent analyzes if they are not robust. On the other hand, if the data is only involved in analyses which require stationarity but are robust against outliers, it would be of interest to have robust change point procedures such as in \cite{dehling2015robust} for the univariate mean change problem. At this point it should be mentioned that even though for the AMOC alternative the null hypothesis is not rejected for sub08992 (see Figure \ref{out.example1} in the supplementary material) the procedures proposed in this work are not robust against outliers as all of them are based on the empirical covariance. For another setting, for example if the outlier occurs rather at the beginning of the observations, the null hypothesis of stationarity might also be rejected for the AMOC alternative which is the case for sub08455 (see Figure \ref{ex08455} in the supplementary material).
\end{itemize}
\begin{figure}[H]
\centering
	\includegraphics[width=0.8\textwidth]{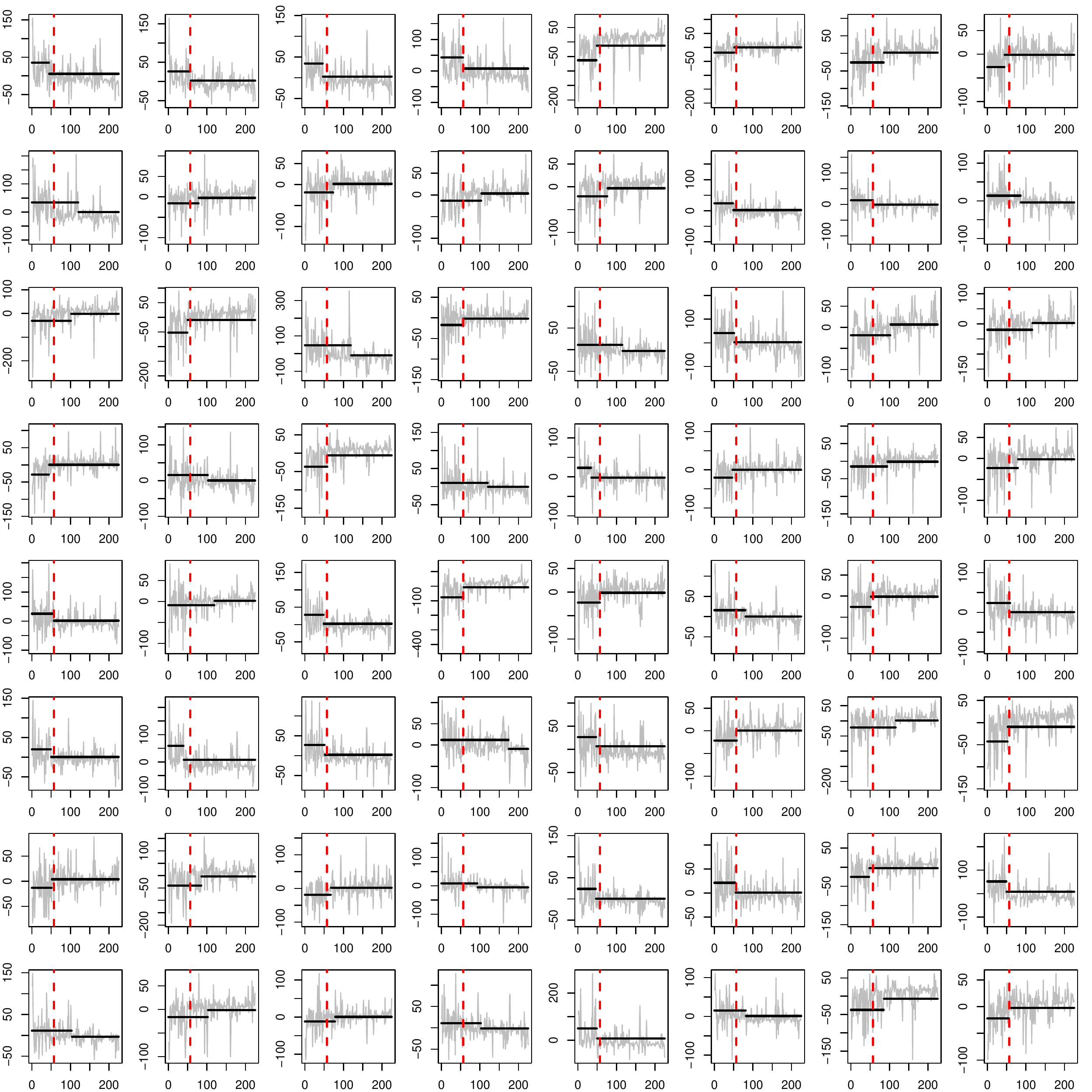}
	\caption{sub12220: 64 score products with the smallest $p$-values for the weighted functional statistic when testing for the AMOC alternative. The global estimated change is $\hat{k}^*=57$ (dashed line).}
	\label{amoc.example1}
\end{figure}
\begin{figure}[H]
\centering
	\includegraphics[width=0.8\textwidth]{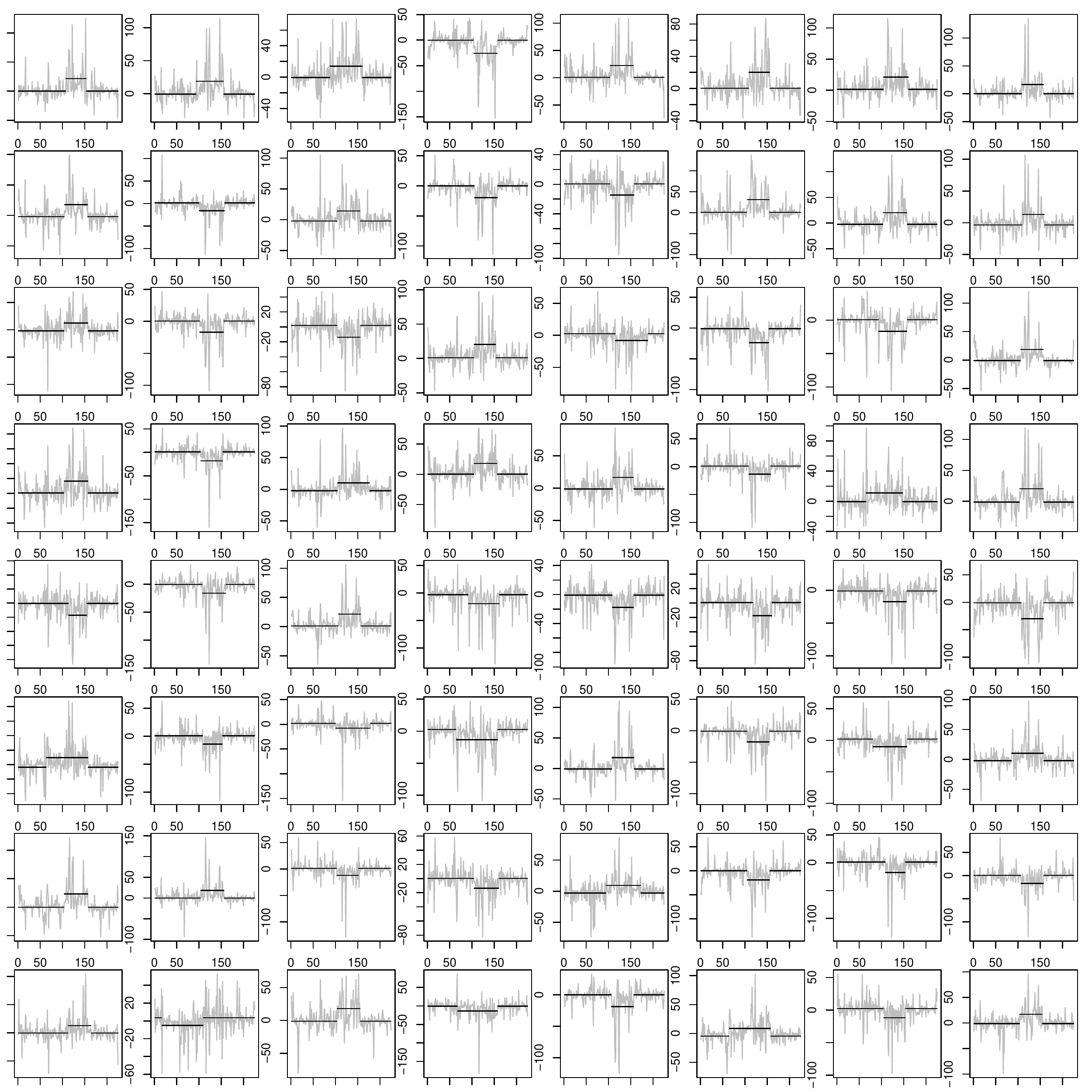}
	\caption{sub08816: 64 score products with the smallest $p$-values for the weighted functional statistic when testing for the epidemic alternative.}
	\label{ep.example2}
\end{figure}
\begin{figure}[H]
\centering
	\includegraphics[width=0.8\textwidth]{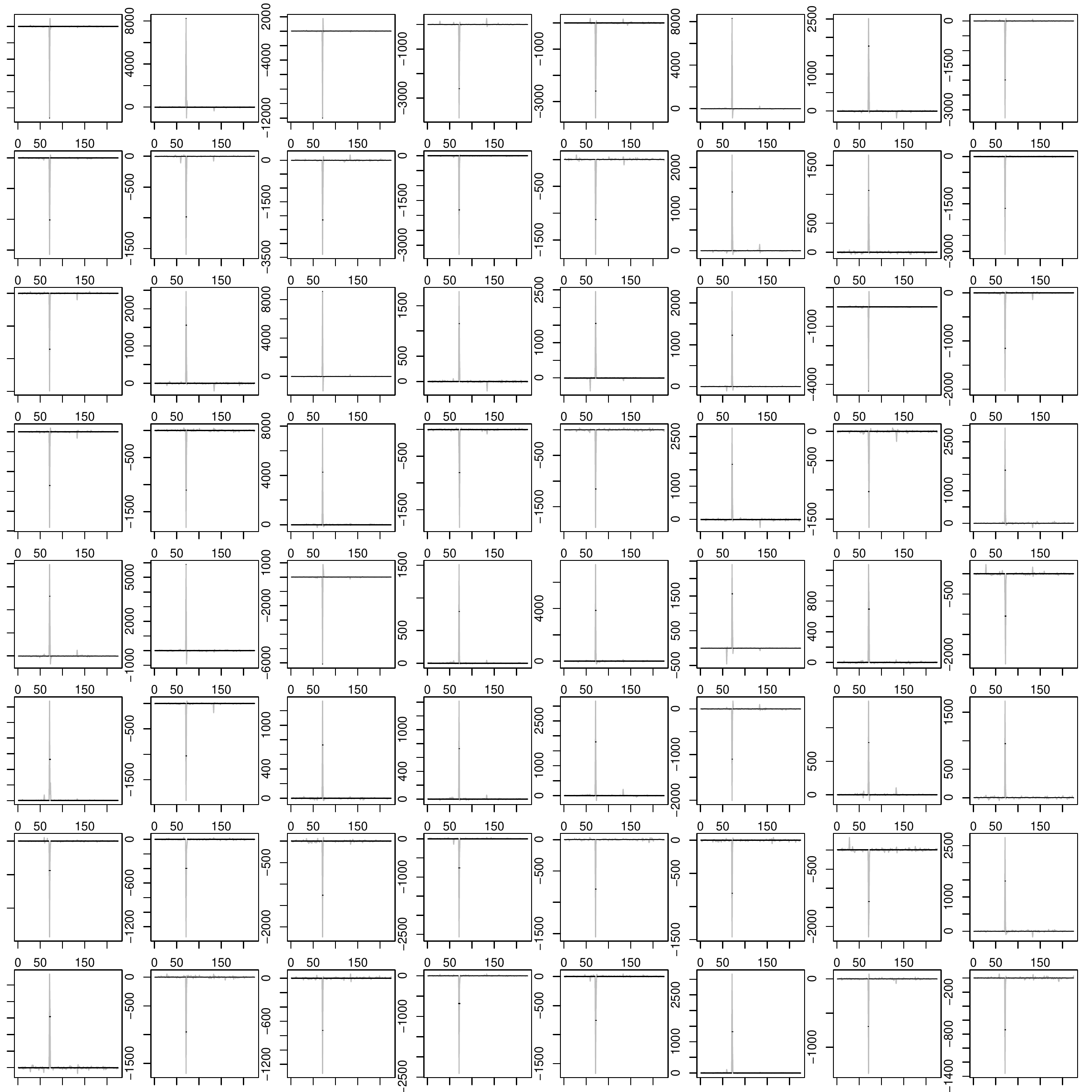}
	\caption{sub08992: 64 score products with the smallest $p$-values for the weighted functional statistic when testing for the epidemic alternative.}
	\label{out.example2}
\end{figure}

\section{Concluding remarks}\label{Conclusions}
In this paper, different methods for detecting devtiations from stationarity in the covariance structure of functional time series have been introduced and investigated with the main focus on applications to fMRI data. Dimension reduction via projections is a very common approach in functional time series analysis and enables the application of a multivariate change point procedure. We derived the asymptotic distribution of the test statistic based on the projection scores for the AMOC alternative as well as for the epidemic alternative. This asymptotic procedure requires the estimation of the long-run covariance which is statistically unstable but can be avoided by using resampling procedures. We applied a circular block bootstrap to obtain the critical values for an adapted test statistic where we only correct for the diagonal elements of the long-run covariance. This gave us a reasonable approach for detecting changes in the covariance structure of fMRI data which, however, comes with the risk of missing changes that are orthogonal to the projection subspace. As alternative solution we provided two test statistics which both take the full functional structure into account and differ with respect to the weighting. The unweighted functional test statistic has been derived from the $L^2$-norm of the functional partial sum process without additional weights. In contrast to that, the weights in the multivariate procedure correct for differenct variances of the components. We incorporate this idea into the functional approach by proposing the weighted functional test statistic. Simulations confirmed that this statistic indeed improves the unweighted functional procedure in different situations and is thus a very promising approach for the detection of change points in functional data analysis, not only for detecting changes in the covariance as considered in this paper but, in an analogous version, also for the mean change problem. A mathematical investigation of this test statistic, as for example deriving the asymptotic distribution, will be of future interest. While the validity of the multivariate block bootstrap has been proven in \cite{Weber}, it still has to be shown for the functional procedures. However, the simulation study already indicates their reasonable performance. The application of the proposed methods to resting state fMRI data has shown that taking possible nonstationarities in the covariance structure into account is crucial. Although we only considered data sets where no mean change was detected the null hypothesis of covariance stationarity was still rejected in more than one third of the cases. Many of those nonstationarities have been detected when testing for the AMOC alternative while in some cases the epidemic alternative seemed to be more appropriate. For some data sets, the null hypothesis was rejected due to outliers, so that the development of more robust methods is of future interest.

\subsection*{Acknowledgements}
This work was supported by the Deutsche Forschungsgemeinschaft (DFG, German Research Foundation) - 314838170, GRK 2297 MathCoRe. The first author would like to thank the German Academic Exchange Service (DAAD) for supporting her visit in Cambridge. The second author was supported by the Engineering and Physical Sciences Research Council (UK) grants  EP/K021672/2, EP/N031938/1 as well as EP/N014588/1. The authors would like to thank the Isaac Newton Institute for Mathematical Sciences for support and hospitality during the programme Statistical Scalability (supported by EPSRC grant numbers EP/K032208/1 and EP/R014604/1).\\
The authors (especially JA) thank Wenda Zhou (Columbia) for his helpful discussions, particularly of his independent derivation of the change point statistics using PCA for covariances under iid settings that were produced as part of his Cambridge Part III essay. 

\bibliography{FdCov_arxiv}
\appendix
\part*{Supplementary material}
This supplement contains additional technical details, proofs and further results from the data analysis.
\section{Max-type test statistics}\label{ap.max}
An alternative to the sum-type statistics discussed in the main paper are the following max-type statistics. For the procedure based on dimension reduction such a statistic is given by $$\Lambda_n=\max_{1\leq k\leq n}S_k^T\hat{\Sigma}_n^{-1}S_k$$ for the AMOC-alternative and $$\Lambda^{ep}_n=\max_{1\leq k_1<k_2\leq n}S_{k_1,k_2}^T\hat{\Sigma}_n^{-1}S_{k_1,k_2}$$ for the epidemic alternative. The asymptotic distributions under the null hypothesis are stated in the following theorem.
\begin{theorem}\label{nullas.max}
Let $\{Y_t(\cdot)\}$ be $L_m^4-$approximable with $\E\|Y_1(\cdot)\|^4<\infty$.
Then, the following asymptotics hold under the null hypothesis if $\hat{\Sigma}$ is a consistent estimator for the long-run covariance $\Sigma$:
\begin{align*}
\Lambda_n&\stackrel{\mathcal D}{\rightarrow}\sup_{0\leq x\leq1}\sum_{l=1}^{\mathfrak{d}}B_l^2(x)\\
&\mbox{as well as}\\
\Lambda^{ep}_n&\stackrel{\mathcal D}{\rightarrow}\sup_{0\leq x<y \leq1}\sum_{l=1}^{\mathfrak{d}}\left(B_l(y)-B_l(x)\right)^2
\end{align*}
where \(\mathfrak{d}=d(d+1)/2\) and \((B_l(x):x\in [0,1],1\leq l\leq\mathfrak{d})\) are independent standard Brownian bridges.
\end{theorem}
\section{Proofs}\label{proofs}
\begin{theorem}\label{prodv.cons}
Let $\hat{v}_l(\cdot)$ be orthonormal eigenfunctions of $\hat{c}_n(u,s)$ and $\tilde{v}_l(\cdot)$ be orthonormal eigenfunctions of $\tilde{c}(u,s)$, where both sets of eigenfunctions are arranged according to the respective eigenvalues in decreasing order. Furthermore, assume that the eigenvalues of $\tilde{c}(u,s)$ are separated, i.e. $\tilde{\lambda}_1>\tilde{\lambda}_2>\ldots>\tilde{\lambda}_d>\tilde{\lambda}_{d+1}$.
\begin{itemize}
\item[a)] If $\int\int (\hat{c}_n(u,s)-\tilde{c}(u,s))^2du\,ds=o_P(1)$, it holds for $l_1,l_2=1,\ldots,d$
$$\int\int\left(\tilde{g}_{l_1}\tilde{g}_{l_2}\hat{v}_{l_1}(u)\hat{v}_{l_2}(s)-\tilde{v}_{l_1}(u)\tilde{v}_{l_2}(s)\right)^2du\,ds=o_P(1),$$
where $\tilde{g}_l=sgn\left(\int\tilde{v}_l(s)\hat{v}_l(s)ds\right)$.
\item[b)] If $\int\int (\hat{c}_n(u,s)-\tilde{c}(u,s))^2du\,ds=O_P(n^{-1})$, it holds for $l_1,l_2=1,\ldots,d$
$$\int\int\left(\tilde{g}_{l_1}\tilde{g}_{l_2}\hat{v}_{l_1}(u)\hat{v}_{l_2}(s)-\tilde{v}_{l_1}(u)\tilde{v}_{l_2}(s)\right)^2du\,ds=O_P(n^{-1}).$$
\end{itemize}
\end{theorem}
\begin{proof}
Observing that
\begin{align*}
&\int\int\left(\tilde{g}_{l_1}\tilde{g}_{l_2}\hat{v}_{l_1}(u)\hat{v}_{l_2}(s)-\tilde{v}_{l_1}(u)\tilde{v}_{l_2}(s)\right)^2du\,ds\notag\\
=&\int\int\left((\tilde{g}_{l_1}\hat{v}_{l_1}(u)-\tilde{v}_{l_1}(u))(\tilde{g}_{l_2}\hat{v}_{l_2}(s)-\tilde{v}_{l_2}(s))+\tilde{v}_{l_1}(u)(\tilde{g}_{l_2}\hat{v}_{l_2}(s)-\tilde{v}_{l_2}(s))\right.\\
&\left.+\tilde{v}_{l_2}(s)(\tilde{g}_{l_1}\hat{v}_{l_1}(u)-\tilde{v}_{l_1}(u))\right)^2du\,ds\notag\\
\leq&\, C\left(\int(\tilde{g}_{l_1}\hat{v}_{l_1}(u)-\tilde{v}_{l_1}(u))^2du\int(\tilde{g}_{l_2}\hat{v}_{l_2}(s)-\tilde{v}_{l_2}(s))^2ds+\int \tilde{v}_{l_1}^2(u)du\int (\tilde{g}_{l_2}\hat{v}_{l_2}(s)-\tilde{v}_{l_2}(s))^2ds\right.\notag\\
&\left.+\int \tilde{v}_{l_2}(s)^2ds\int (\tilde{g}_{l_1}\hat{v}_{l_1}(u)-\tilde{v}_{l_1}(u))^2du\right)\notag\\
=&\, C\left(\int(\tilde{g}_{l_1}\hat{v}_{l_1}(u)-\tilde{v}_{l_1}(u))^2du\int(\tilde{g}_{l_2}\hat{v}_{l_2}(s)-\tilde{v}_{l_2}(s))^2ds+\int (\tilde{g}_{l_2}\hat{v}_{l_2}(s)-\tilde{v}_{l_2}(s))^2ds\right.\\
&\left.+\int (\tilde{g}_{l_1}\hat{v}_{l_1}(u)-\tilde{v}_{l_1}(u))^2du\right)
\end{align*}
the assertions follow with Theorem 2.1 in \cite{fda}.
\end{proof}
\begin{proof}[Proof of Theorem \ref{nullas} and \ref{nullas.max}]
We first show that the $L_m^p-$approximability is passed on to the projection scores. Let $Y_t^{(m)}$ be the $m$-approximations for an  $L_m^p-$approximable sequence $Y_t$. The sequence $\eta_t^{(m)}$ with components $\eta_{t,l}^{(m)}=\int Y_t^{(m)}(s)v_l(s)ds$ is m-dependent as $Y_t^{(m)}$ is $m$-dependent. Furthermore, it holds with the Cauchy-Schwarz inequality, similarly to the proof of Theorem 5.1 in \cite{HK},
\begin{align*}
&\sum_{m\geq 1}\left(\E\left[\left|\eta_t-\eta_t^{(m)}\right|^p\right]\right)^{1/p}=\sum_{m\geq 1}\left(\E\left[\left(\sum_{l=1}^d\left(\eta_{t,l}-\eta_{t,l}^{(m)}\right)^2\right)^{p/2}\right]\right)^{1/p}\\
=&\sum_{m\geq 1}\left(\E\left[\left(\sum_{l=1}^d\left(\int\left(Y_t(s)-Y_t^{(m)}(s)\right)v_l(s)ds\right)^2\right)^{p/2}\right]\right)^{1/p}\\
\leq&\sum_{m\geq 1}\left(\E\left[\left(\sum_{l=1}^d\int\left(Y_t(s)-Y_t^{(m)}(s)\right)^2ds\int v_l^2(s)ds\right)^{p/2}\right]\right)^{1/p}\\
=&\sqrt{d}\sum_{m\geq 1}\left(\E\left[\left(\int\left(Y_t(s)-Y_t^{(m)}(s)\right)^2ds\right)^{p/2}\right]\right)^{1/p}=\sqrt{d}\sum_{m\geq 1}\left(\E\left[\left\| Y_t-Y_t^{(m)}\right\| ^{p}\right]\right)^{1/p}<\infty,
\end{align*}
where $|\cdot|$ denotes the Euklidean norm. Thus, the score vectors  $\{\eta_t\}_{t\geq 1}$ with components $\eta_{t,l}=\int Y_t(s)v_l(s)ds$ are $L_m^4-$approximable and it holds with Theorem A.2 in \cite{Aue} and the continuous mapping theorem
\begin{align}\label{limeta}
\frac{1}{\sqrt{n}}\left(\sum_{t=1}^{[nx]}\vech[\eta_t\eta_t^T]-\frac{[nx]}{n}\sum_{t=1}^n\vech[\eta_t\eta_t^T]\right)\stackrel{D^{\mathfrak{d}}[0,1]}{\rightarrow}B_{\Sigma}(x),
\end{align}
where $\{B_{\Sigma}(x):0\leq x\leq 1\}$ is a $\mathfrak{d}$-dimensional centered Gaussian process with covariance function $\Cov\left(B_{\Sigma}(x),B_{\Sigma}(y)\right)=\Sigma(\min\{x,y\}-xy)$. The convergence in (\ref{limeta}) still holds true if the projection basis is obtained based on the empirical covariance kernel. More precisely, it holds for $\check{\eta}_{t,l}=\int Y_t(s)\hat{v}_l(s)ds$ and $g_l=sgn\left(\int v_l(s)\hat{v}_l(s)ds\right)$ with the Cauchy-Schwarz inequality
\begin{align}\label{limeta.approx1a}
&\sup_{0\leq x\leq 1}\left|\frac{1}{\sqrt{n}}\sum_{t=1}^{[nx]}g_{l_1}g_{l_2}\left(\check{\eta}_{t,l_1}\check{\eta}_{t,l_2}-\overline{\check{\eta}_{l_1}\check{\eta}_{l_2}}\right)-\frac{1}{\sqrt{n}}\sum_{t=1}^{[nx]}\left(\eta_{t,l_1}\eta_{t,l_2}-\overline{\eta_{l_1}\eta_{l_2}}\right)\right|\notag\\
=&\sup_{0\leq x\leq 1}\left|\int\int\left(\frac{1}{n}\sum_{t=1}^{[nx]}\left(Y_t(u)Y_t(s)-\overline{Y(u)Y(s)}\right)\right)\sqrt{n}\left(g_{l_1}g_{l_2}\hat{v}_{l_1}(u)\hat{v}_{l_2}(s)-v_{l_1}(u)v_{l_2}(s)\right)du\,ds\right|\notag\\
\leq&\sup_{0\leq x\leq 1}\left(\int\int \left(\frac{1}{n}\sum_{t=1}^{[nx]}\left(Y_t(u)Y_t(s)-\overline{Y(u)Y(s)}\right)\right)^2du\,ds\right)^{\frac 12}\left(n\int\int\left(g_{l_1}g_{l_2}\hat{v}_{l_1}(u)\hat{v}_{l_2}(s)\right.\right.\notag\\
&\left.\phantom{\int}\left.-v_{l_1}(u)v_{l_2}(s)\right)^2du\,ds\right)^{\frac 12}.
\end{align}
By Lemma 2.3 b) in \cite{fda} and the separation of the eigenvalues of $c(u,s)$ the assumptions of Theorem  \ref{prodv.cons} b) are fulfilled such that we obtain 
$$\left(n\int\int\left(g_{l_1}g_{l_2}\hat{v}_{l_1}(u)\hat{v}_{l_2}(s)-v_{l_1}(u)v_{l_2}(s)\right)^2du\,ds\right)^{\frac 12}=O_P(1).$$
Lemma 2.1 in \cite{HK} yields that $Z_t(u,s)=Y_t(u)Y_t(s)$ is $L_m^4-$approximable and with the invariance principle in \cite{Berkinv} we obtain
\begin{align}\label{inv.prod}
\sup_{0\leq x\leq 1}\int\int\left(\frac{1}{n}\sum_{t=1}^{[nx]}\left[Y_t(u)Y_t(s)-E(Y_1(u)Y_1(s))\right]\right)^2du\,ds=O_P\left(n^{-1}\right)=o_P(1).
\end{align}
It follows that
\begin{align*}
&\sup_{0\leq x\leq 1}\int\int\left(\frac{1}{n}\sum_{t=1}^{[nx]}\left(Y_t(u)Y_t(s)-\overline{Y(u)Y(s)}\right)\right)^2du\,ds\notag\\
=&\sup_{0\leq x\leq 1}\int\int\left(\left(\frac{1}{n}\sum_{t=1}^{[nx]}\left[Y_t(u)Y_t(s)-E\left(Y_1(u)Y_1(s)\right)\right]\right)\right.\\
&\quad\left.-\frac{[nx]}{n}\left(\frac 1n\sum_{t=1}^n\left[Y_t(u)Y_t(s)-E\left(Y_1(u)Y_1(s)\right)\right]\right)\right)^2du\,ds\notag\\
\leq&\,C\sup_{0\leq x\leq 1}\int\int\left(\frac{1}{n}\sum_{t=1}^{[nx]}\left[Y_t(u)Y_t(s)-E\left(Y_1(u)Y_1(s)\right)\right]\right)^2du\,ds\\
&\quad+C\int\int\left(\frac 1n\sum_{t=1}^n\left[Y_t(u)Y_t(s)-E\left(Y_1(u)Y_1(s)\right)\right]\right)^2du\,ds\notag\\
\leq&\,2C\sup_{0\leq x\leq 1}\int\int\left(\frac{1}{n}\sum_{t=1}^{[nx]}\left[Y_t(u)Y_t(s)-E\left(Y_1(u)Y_1(s)\right)\right]\right)^2du\,ds=o_P(1).
\end{align*}
Hence, (\ref{limeta.approx1a}) yields
\begin{align}\label{limeta.approx1}
&\sup_{0\leq x\leq 1}\left|\frac{1}{\sqrt{n}}\sum_{t=1}^{[nx]}g_{l_1}g_{l_2}\left(\check{\eta}_{t,l_1}\check{\eta}_{t,l_2}-\overline{\check{\eta}_{l_1}\check{\eta}_{l_2}}\right)-\frac{1}{\sqrt{n}}\sum_{t=1}^{[nx]}\left(\eta_{t,l_1}\eta_{t,l_2}-\overline{\eta_{l_1}\eta_{l_2}}\right)\right|=o_P(1).
\end{align}
We obtain the same limit distribution if we replace $\check{\eta}_{t,l}$ by $\hat{\eta}_{t,l}=\int \left(X_t(s)-\overline{X}_n(s)\right)\hat{v}_l(s)ds$ as in our statistics. Indeed, with the notations $\tilde{Y}_t:=Y_t-\overline{Y}_n,$ $\overline{\tilde{Y}(u)\tilde{Y}(s)}=\frac 1n\sum_{t=1}^n\tilde{Y}_t(u)\tilde{Y}_t(s),$ $\overline{Y}_{k}(u)=\frac{1}{n}\sum_{t=1}^{k}Y_t(u)$ we obtain
\begin{align}\label{limeta.approx2}
&\sup_{0\leq x\leq 1}\left|\frac{1}{\sqrt{n}}\sum_{t=1}^{[nx]}\left(\check{\eta}_{t,l_1}\check{\eta}_{t,l_2}-\overline{\check{\eta}_{l_1}\check{\eta}_{l_2}}\right)-\frac{1}{\sqrt{n}}\sum_{t=1}^{[nx]}\left(\hat{\eta}_{t,l_1}\hat{\eta}_{t,l_2}-\overline{\hat{\eta}_{l_1}\hat{\eta}_{l_2}}\right)\right|\notag\\
=&\sup_{0\leq x\leq 1}\left|\frac{1}{\sqrt{n}}\sum_{t=1}^{[nx]}\int\int \left(Y_t(u)Y_t(s)-\overline{Y(u)Y(s)}\right)\hat{v}_{l_1}(u)\hat{v}_{l_2}(s)du\,ds\right.\notag\\
&\left.-\frac{1}{\sqrt{n}}\sum_{t=1}^{[nx]}\int\int \left(\tilde{Y}_t(u)\tilde{Y}_t(s)-\overline{\tilde{Y}(u)\tilde{Y}(s)}\right)\hat{v}_{l_1}(u)\hat{v}_{l_2}(s)du\,ds\right|\notag\\
=&\sup_{0\leq x\leq 1}\left|\int\int\left(\frac{1}{\sqrt{n}}\sum_{t=1}^{[nx]}\left(Y_t(u)\overline{Y}_n(s)+\overline{Y}_n(u)Y_t(s)-2\overline{Y}_n(u)\overline{Y}_n(s)\right)\right)\hat{v}_{l_1}(u)\hat{v}_{l_2}(s)du\,ds\right|\notag\\
=&\sup_{0\leq x\leq 1}\left|\int\overline{Y}_{[nx]}(u)\hat{v}_{l_1}(u)du\int\sqrt{n}\,\overline{Y}_{n}(s)\hat{v}_{l_2}(s)ds+\int\sqrt{n}\,\overline{Y}_{n}(u)\hat{v}_{l_1}(u)du\int\overline{Y}_{[nx]}(s)\hat{v}_{l_2}(s)ds\right.\notag\\*
&\left.\quad-2\frac{[nx]}{n}\int\sqrt{n}\,\overline{Y}_{n}(u)\hat{v}_{l_1}(u)du\int\overline{Y}_{n}(s)\hat{v}_{l_2}(s)ds\right|\notag\\
\leq&\sup_{0\leq x\leq 1}\left(\left|\int\overline{Y}_{[nx]}(u)\hat{v}_{l_1}(u)du\int\sqrt{n}\,\overline{Y}_{n}(s)\hat{v}_{l_2}(s)ds\right|+\left|\int\sqrt{n}\,\overline{Y}_{n}(u)\hat{v}_{l_1}(u)du\int\overline{Y}_{[nx]}(s)\hat{v}_{l_2}(s)ds\right|\right.\notag\\*
&\quad\left.+\left|2\frac{[nx]}{n}\int\sqrt{n}\,\overline{Y}_{n}(u)\hat{v}_{l_1}(u)du\int\overline{Y}_{n}(s)\hat{v}_{l_2}(s)ds\right|\right)\notag\\
\leq&\sup_{0\leq x\leq 1}\left[\left(\int\overline{Y}^2_{[nx]}(u)du\right)^{\frac 12}\left(\int\hat{v}^2_{l_1}(u)du\right)^{\frac 12}\left(\int\left(\sqrt{n}\,\overline{Y}_{n}(s)\right)^2ds\right)^{\frac 12}\left(\int\hat{v}^2_{l_2}(s)ds\right)^{\frac 12}\right.\notag\\*
&\quad+\left(\int\left(\sqrt{n}\,\overline{Y}_{n}(u)\right)^2du\right)^{\frac 12}\left(\int\hat{v}^2_{l_1}(u)du\right)^{\frac 12}\left(\int\overline{Y}^2_{[nx]}(s)ds\right)^{\frac 12}\left(\int\hat{v}^2_{l_2}(s)ds\right)^{\frac 12}\notag\\*
&\quad\left.+2\frac{[nx]}{n}\left(\int\left(\sqrt{n}\,\overline{Y}_{n}(u)\right)^2du\right)^{\frac 12}\left(\int\hat{v}^2_{l_1}(u)du\right)^{\frac 12}\left(\int\overline{Y}^2_{n}(s)ds\right)^{\frac 12}\left(\int\hat{v}^2_{l_2}(s)ds\right)^{\frac 12}\right]\notag\\
=&\sup_{0\leq x\leq 1}2\left(\int\overline{Y}^2_{[nx]}(u)du\right)^{\frac 12}\left(\int\left(\sqrt{n}\,\overline{Y}_{n}(s)\right)^2ds\right)^{\frac 12}\notag\\
&+2\sup_{0\leq x\leq 1}\frac{[nx]}{n}\left(\int\left(\sqrt{n}\,\overline{Y}_{n}(u)\right)^2du\right)^{\frac 12}\left(\int\overline{Y}^2_{n}(s)ds\right)^{\frac 12}\notag\\
\leq&2\left(\int\left(\sqrt{n}\,\overline{Y}_{n}(s)\right)^2ds\right)^{\frac 12}\left(\left(\sup_{0\leq x\leq 1}\int\overline{Y}^2_{[nx]}(s)ds\right)^{\frac 12}+\left(\int\overline{Y}^2_{n}(s)ds\right)^{\frac 12}\right)=o_P(1)
\end{align}
as it holds with the ergodic theorem (see, for example, \cite{rao})
\begin{align}\label{yinv2}
\int\overline{Y}^2_{n}(s)ds=o_P(1).
\end{align}
Combining (\ref{limeta}), (\ref{limeta.approx1}) and (\ref{limeta.approx2}) we obtain
\begin{align*}
S_{[nx]}\stackrel{D^{\mathfrak{d}}[0,1]}{\rightarrow}B_{\Sigma}(x).
\end{align*}
and the assertions follow by the continuous mapping theorem.
\end{proof}
\paragraph*{Behaviour under alternatives}
\begin{proof}[Proof of Lemma \ref{climit}]
We split the empirical covariance as follows:
\begin{align*}
&\hat{c}_n(u,s)=\frac1n\sum_{t=1}^n\left(X_t(u)-\overline{X}_n(u)\right)\left(X_t(s)-\overline{X}_n(s)\right)=\frac1n\sum_{t=1}^n\left(Y_t(u)-\overline{Y}_n(u)\right)\left(Y_t(s)-\overline{Y}_n(s)\right)\notag\\
=&\frac{1}{n}\sum_{t=1}^{[\theta n]}\left(Y_t^{(1)}(u)-\overline{Y}_n(u)\right)\left(Y_t^{(1)}(s)-\overline{Y}_n(s)\right)+\frac{1}{n}\sum_{t=[\theta n]+1}^n\left(Y_t^{(2)}(u)-\overline{Y}_n(u)\right)\left(Y_t^{(2)}(s)-\overline{Y}_n(s)\right).
\end{align*}
Now, observe that
\begin{align*}
&\int\int\left(\frac{1}{n}\sum_{t=1}^{[\theta n]}\left[\left(Y_t^{(1)}(u)-\overline{Y}_n(u)\right)\left(Y_t^{(1)}(s)-\overline{Y}_n(s)\right)- c(u,s)\right]\right)^2du\,ds\notag\\
\leq&\, C\int\int\left(\frac{1}{n}\sum_{t=1}^{[\theta n]}\left(Y_t^{(1)}(u)Y_t^{(1)}(s)-E(Y_1(u)Y_1(s))\right)\right)^2du\,ds\notag\\*
&\,+C\int\int\left((\theta+o(1))\overline{Y}_n(u)\overline{Y}_n(s)-\overline{Y}_{[\theta n]}(u)\overline{Y}_n(s)-\overline{Y}_n(u)\overline{Y}_{[\theta n]}(s)\right)^2du\,ds.
\end{align*}
Furthermore, it holds 
\begin{align*}
&\int\int\left((\theta+o(1))\overline{Y}_n(u)\overline{Y}_n(s)-\overline{Y}_{[\theta n]}(u)\overline{Y}_n(s)-\overline{Y}_n(u)\overline{Y}_{[\theta n]}(s)\right)^2du\,ds\notag\\
\leq&\,C\left((\theta+o(1))\int\overline{Y}^2_n(u)du\int\overline{Y}^2_n(s)ds+\int\overline{Y}^2_{[\theta n]}(u)du\int\overline{Y}^2_n(s)ds+\int\overline{Y}^2_n(u)du\int\overline{Y}^2_{[\theta n]}(s)ds\right)\\
=&\,o_P(1)
\end{align*}
by (\ref{yinv2}), where one needs to note that this assertion remains true under the alternative which can easily be seen by splitting the time series at the change point. By the ergodic theorem it holds
\begin{align*}
&\int\int\left(\frac{1}{n}\sum_{t=1}^{[\theta n]}\left(Y_t^{(1)}(u)Y_t^{(1)}(s)-E(Y_1(u)Y_1(s))\right)\right)^2du\,ds\\
=&\left(\frac{[\theta n]}{n}\right)^2\int\int\left(\frac{1}{[\theta n]}\sum_{t=1}^{[\theta n]}\left(Y_t^{(1)}(u)Y_t^{(1)}(s)-E(Y_1(u)Y_1(s))\right)\right)^2du\,ds=o_P(1).
\end{align*}
Hence, we obtain
\begin{align*}
&\int\int\left(\frac{1}{n}\sum_{t=1}^{[\theta n]}\left[\left(Y_t^{(1)}(u)-\overline{Y}_n(u)\right)\left(Y_t^{(1)}(s)-\overline{Y}_n(s)\right)-c(u,s)\right]\right)^2du\,ds=o_P(1)
\end{align*}
and analogously
\begin{align*}
&\int\int\left(\frac{1}{n}\sum_{t=[\theta n]+1}^n\left[\left(Y_t^{(2)}(u)-\overline{Y}_n(u)\right)\left(Y_t^{(2)}(s)-\overline{Y}_n(s)\right)-(c(u,s)+\delta(u,s))\right]\right)^2du\,ds\\
=&\,o_P(1).
\end{align*}
As $$\int\int\left(\frac{[\theta n]}{n}c(u,s)+\frac{n-[\theta n]}{n}\left(c(u,s)+\delta(u,s)\right)-k(u,s)\right)^2du\,ds=o_P(1),$$
where $k(u,s)=\theta c(u,s)+(1-\theta)\left(c(u,s)+\delta(u,s)\right)=c(u,s)+(1-\theta)\delta(u,s)$, it follows that
\begin{align}\label{cconv}
\int\int \left(\hat{c}_n(u,s)-k(u,s)\right)^2du\,ds=o_P(1).
\end{align}

\end{proof}
\paragraph*{Example 1}
In this setting, condition (\ref{changedet}) is fufilled as it holds
\begin{align}\label{eigenval}
\int\delta(u,s)v_{l}(s)ds=\int\left(c(u,s)+\delta(u,s)\right)v_{l}(s)ds-\int c(u,s)v_{l}(s)ds=\delta_{l}v_{l}(u)
\end{align}
and thus
\begin{equation}
\begin{split}\label{ex.1}
\int\int\delta(u,s)v_{l_1}(u)v_{l_2}(s)du\,ds=\delta_{l_1}\int v_{l_1}(u)v_{l_2}(u)du=\begin{cases}
0,&l_1\neq l_2\\
\delta_{l_1},& l_1=l_2.
\end{cases}
\end{split}
\end{equation}
By (\ref{eigenval}) each $v_{l}$ is an eigenfunction of $k(u,s)$ with eigenvalue $\lambda_{l}+\theta\delta_{l}$. It follows with the Cauchy-Schwarz inequality
\begin{align}\label{vconv}
&\left|g_{l_1}g_{l_2}\int\int\delta(u,s)v_{l_1}(u)v_{l_2}(s)du\,ds-\int\int\delta(u,s)\hat{v}_{l_1}(u)\hat{v}_{l_2}(s)du\,ds \right|\notag\\
=&\left|\int\int\delta(u,s)\left(g_{l_1}g_{l_2}v_{l_1}(u)v_{l_2}(s)-\hat{v}_{l_1}(u)\hat{v}_{l_2}(s)\right)du\,ds \right|\notag\\
\leq&\left(\int\int\delta^2(u,s)du\,ds\right)^{\frac 12}\left(\int\int\left(g_{l_1}g_{l_2}v_{l_1}(u)v_{l_2}(s)-\hat{v}_{l_1}(u)\hat{v}_{l_2}(s)\right)^2du\,ds\right)^{\frac 12}=o_p(1)
\end{align}
with Theorem \ref{prodv.cons} a) and $\delta(u,s)\in \mathcal{L}^2(\mathcal{Z})$. Hence, we get
\begin{align}\label{ex1.del}
&\int\int\delta(u,s)\hat{v}_{l_1}(u)\hat{v}_{l_2}(s)du\,ds=g_{l_1}g_{l_2}\delta_{l_1}\int v_{l_1}(u)v_{l_2}(u)du+o_P(1)\notag\\
=&\begin{cases}
o_P(1),&l_1\neq l_2\\
g_{l_1}g_{l_2}\delta_{l_1}+o_P(1),& l_1=l_2.
\end{cases}
\end{align}
This shows that the change is detectable if the eigendirections are estimated based on the empirical covariance function.
\paragraph*{Example 2}
First observe that, as $\epsilon_{t,l}$ is independent of $\eta_{t,l}$ and as the score components are uncorrelated,
$$\Cov(\eta_{t,k}+\epsilon_{t,k},\eta_{t,l}+\epsilon_{t,l})=\Cov(\eta_{t,k},\eta_{t,l})+\Cov(\epsilon_{t,k},\epsilon_{t,l})=\begin{cases}
\lambda_{k}+\sigma_{k,k},&k=l,\\
\sigma_{k,l},&k\neq l.
\end{cases}
$$
Hence, it holds with (\ref{KL}) for $t>\theta n$
\begin{align*}
\Cov(X_t(u),X_t(s))=&\sum_{k,l=1}^{\infty}v_{k}(u)v_{l}(s)\Cov(\tilde{\eta}_{t,k},\tilde{\eta}_{t,l})\\
=&\sum_{l=m+1}^{\infty}\lambda_l v_{l}(u)v_{l}(s)+\sum_{k,l=1}^{m}v_{k}(u)v_{l}(s)\Cov(\eta_{t,k}+\epsilon_{t,k},\eta_{t,l}+\epsilon_{t,l})\notag\\
=&\sum_{l=m+1}^{\infty}\lambda_l v_{l}(u)v_{l}(s)+\sum_{l=1}^m(\lambda_{l}+\sigma_{l,l})v_{l}(u)v_{l}(s)+\sum_{k,l=1,k\neq l}^{m}\sigma_{k,l}v_{k}(u)v_{l}(s)\\
=&\sum_{l=1}^{\infty}\lambda_l v_{l}(u)v_{l}(s)+\sigma_{k,l} \sum_{k,l=1}^{m}v_{k}(u)v_{l}(s)\notag\\
=&c(u,s)+1_{\{\theta n<t\leq n\}}\sum_{k,l=1}^{m}\sigma_{k,l}v_{k}(u)v_{l}(s)
\end{align*}
such that the change in the covariance kernel is given by
\begin{align}\label{ex2.delta}
\delta(u,s)=\sum_{k,l=1}^{m}\sigma_{k,l}v_{k}(u)v_{l}(s).
\end{align}
For $l_1,l_2\in\{1,\ldots,m\}$ it holds
\begin{align*}
&\int\int\delta(u,s)v_{l_1}(u)v_{l_2}(s)du\,ds=\sum_{k,l=1}^{m}\sigma_{k,l}\int\int v_{k}(u)v_{l}(s)v_{l_1}(u)v_{l_2}(s)du\,ds\\*
=&\sum_{l,k=1}^{m}\sigma_{k,l}\left(\int v_{k}(u)v_{l_1}(u)du \int v_{l}(s)v_{l_2}(s)ds\right)=\sigma_{l_1,l_2}.
\end{align*}
Hence, condition (\ref{changedet}) is fufilled. Analogously to (\ref{vconv}) we obtain
\begin{align}\label{ex2.est}
\int\int\delta(u,s)\hat{v}_{l_1}(u)\hat{v}_{l_2}(s)du\,ds&=\tilde{g}_l\tilde{g}_k\sum_{k,l=1}^{m}\sigma_{k,l}\left(\int v_{k}(u)\tilde{v}_{l_1}(u)du \int v_{l}(s)\tilde{v}_{l_2}(s)ds\right)+o_P(1),
\end{align}
showing that the change is detectable if the eigendirections are estimated based on the empirical covariance function if $\sum_{k,l=1}^{m}\sigma_{k,l}\left(\int v_{k}(u)\tilde{v}_{l_1}(u)du \int v_{l}(s)\tilde{v}_{l_2}(s)ds\right)\neq 0$ for at least one pair $l_1,l_2\in\{1,\ldots,\min\{d,m\}\}$.
\paragraph*{Long-run covariance for Gaussian scores}
Assuming a normal distribution, the components $\{\eta_{t,l}:l=1,\ldots,d\}$ of the score vectors are independent. This leads to
\begin{align*}
\Cov(\eta_{t,l_1}\eta_{t,l_2},\eta_{t,l_3}\eta_{t,l_4})
=&\begin{cases}
\E(\eta_{t,l_1}^4)-\E(\eta_{t,l_1}^2)^2,&l_1=l_2=l_3=l_4\\
\E(\eta_{t,l_1}^2)\E(\eta_{t,l_3}^2)-\E(\eta_{t,l_1}^2)\E(\eta_{t,l_3}^2),&l_1=l_2\neq l_3=l_4,\\
\E(\eta_{t,l_1}^2)\E(\eta_{t,l_2}^2),&l_1=l_3\neq l_2=l_4,\\
\E(\eta_{t,l_1}^2)\E(\eta_{t,l_2}^2),&l_1=l_4\neq l_2=l_3,\\
0,&\mbox{otherwise},
\end{cases}\\
=&\begin{cases}
3\lambda_{l_1}^2-\lambda_{l_1}^2,&l_1=l_2=l_3=l_4\\
\lambda_{l_1}\lambda_{l_2},&l_1=l_3\neq l_2=l_4,\\
\lambda_{l_1}\lambda_{l_2},&l_1=l_4\neq l_2=l_3,\\
0,&\mbox{otherwise}.
\end{cases}
\end{align*}
With $\vech[\eta_0\eta_0^T]=
          (\eta_{0,1}^2,
          \eta_{0,1}\eta_{0,2},
					\ldots,
					\eta_{0,2}^2,
          \eta_{0,2}\eta_{0,3},
					\ldots,
					\eta_{0,d}^2)
       $, we obtain
$$\Sigma=\Cov\left(\vech[\eta_0\eta_0^T]\right)=\diag(2\lambda_{1}^2,\lambda_{1}\lambda_{2},\ldots,2\lambda_{2}^2,\lambda_{2}\lambda_{3},\ldots,2\lambda_{d}^2).$$
\paragraph*{Functional test statistic}
The representation of the $L^2$-norm of the functional partial sum process in terms of the projection scores as stated in Remark \ref{rem.Tf} is obtained by
\begin{align}\label{l2s}
&\|S_k^F\|^2=\frac 1n\int\int \sum_{t_1,t_2=1}^k\left(\left(X_{t_1}(u)X_{t_1}(s)-\overline{X(u)X(s)}\right)\right)\left(X_{t_2}(u)X_{t_2}(s)-\overline{X(u)X(s)}\right) du\, ds\notag\\
=&\frac 1n \sum_{t_1,t_2=1}^k\sum_{l_1,l_2,l_3,l_4=1}^{\infty}(\eta_{t_1,l_1}\eta_{t_1,l_2}-\overline{\eta_{l_1}\eta_{l_2}})(\eta_{t_2,l_3}\eta_{t_2,l_4}-\overline{\eta_{l_3}\eta_{l_4}})\int v_{l_1}(u)v_{l_3}(u)du \int v_{l_2}(s)v_{l_4}(s)\,ds\notag\\
=&\frac 1n \sum_{t_1,t_2=1}^k\sum_{l_1,l_2=1}^{\infty}(\eta_{t_1,l_1}\eta_{t_1,l_2}-\overline{\eta_{l_1}\eta_{l_2}})(\eta_{t_2,l_1}\eta_{t_2,l_2}-\overline{\eta_{l_1}\eta_{l_2}})=\frac 1n \sum_{l_1,l_2=1}^{\infty}\left(\sum_{t=1}^k(\eta_{t,l_1}\eta_{t,l_2}-\overline{\eta_{l_1}\eta_{l_2}})\right)^2.
\end{align}
\section{FurthersSimulations}\label{sim.app}
Figure \ref{fig.sim} shows the empirical size and the size corrected power where the procedures considered in this paper are applied to the independent innovations of the simulation study in Section \ref{sim} in the main paper using Efron's Bootstrap to obtain the critical values.
\begin{figure}[htb]
\begin{subfigure}[t]{0.32\textwidth}
\centering
	\includegraphics[width=\textwidth]{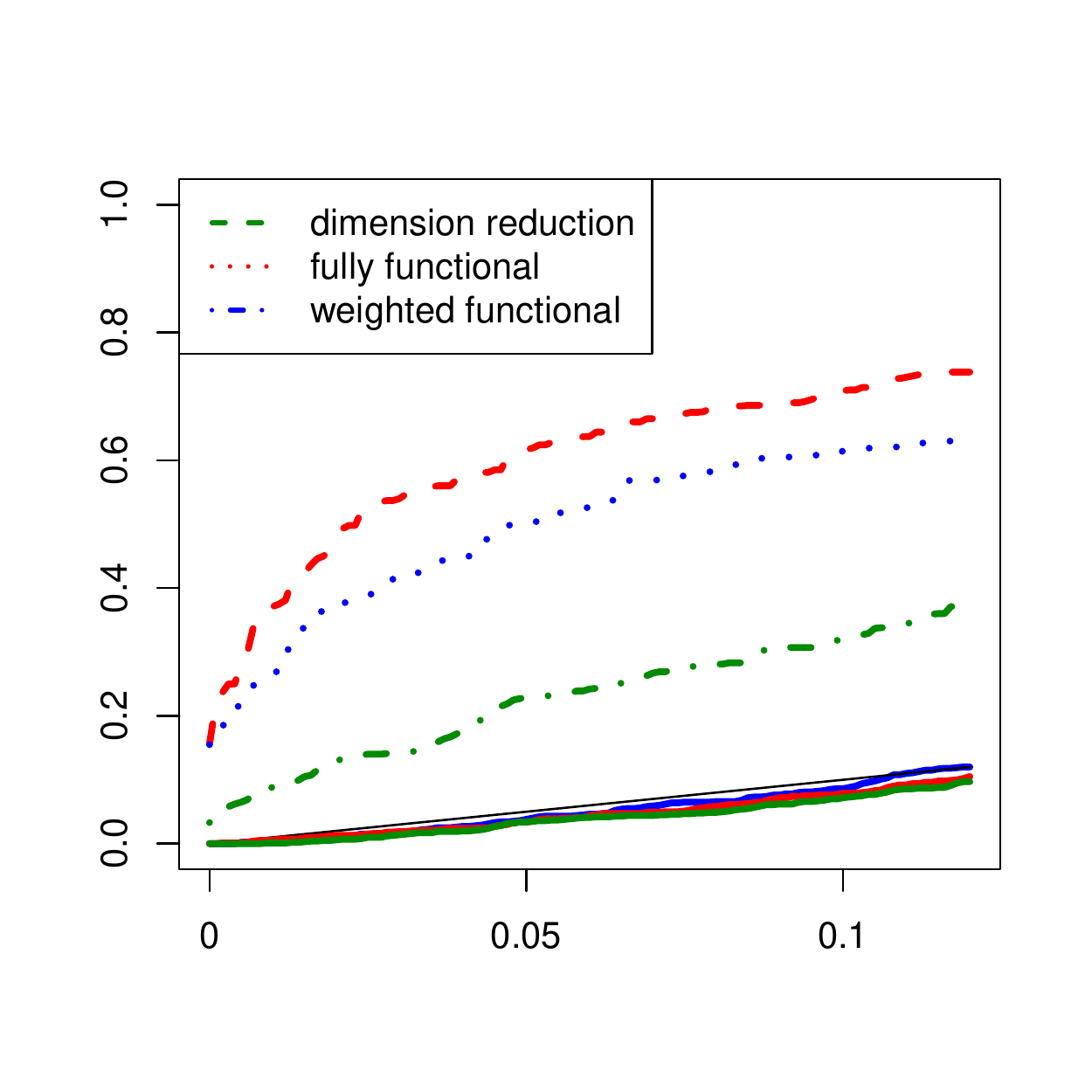}
	\caption{Setting 1,$\sigma_{\epsilon}=1,m=2,25,50$}
	\end{subfigure}
\hfill
\begin{subfigure}[t]{0.32\textwidth}
	\centering
	\includegraphics[width=\textwidth]{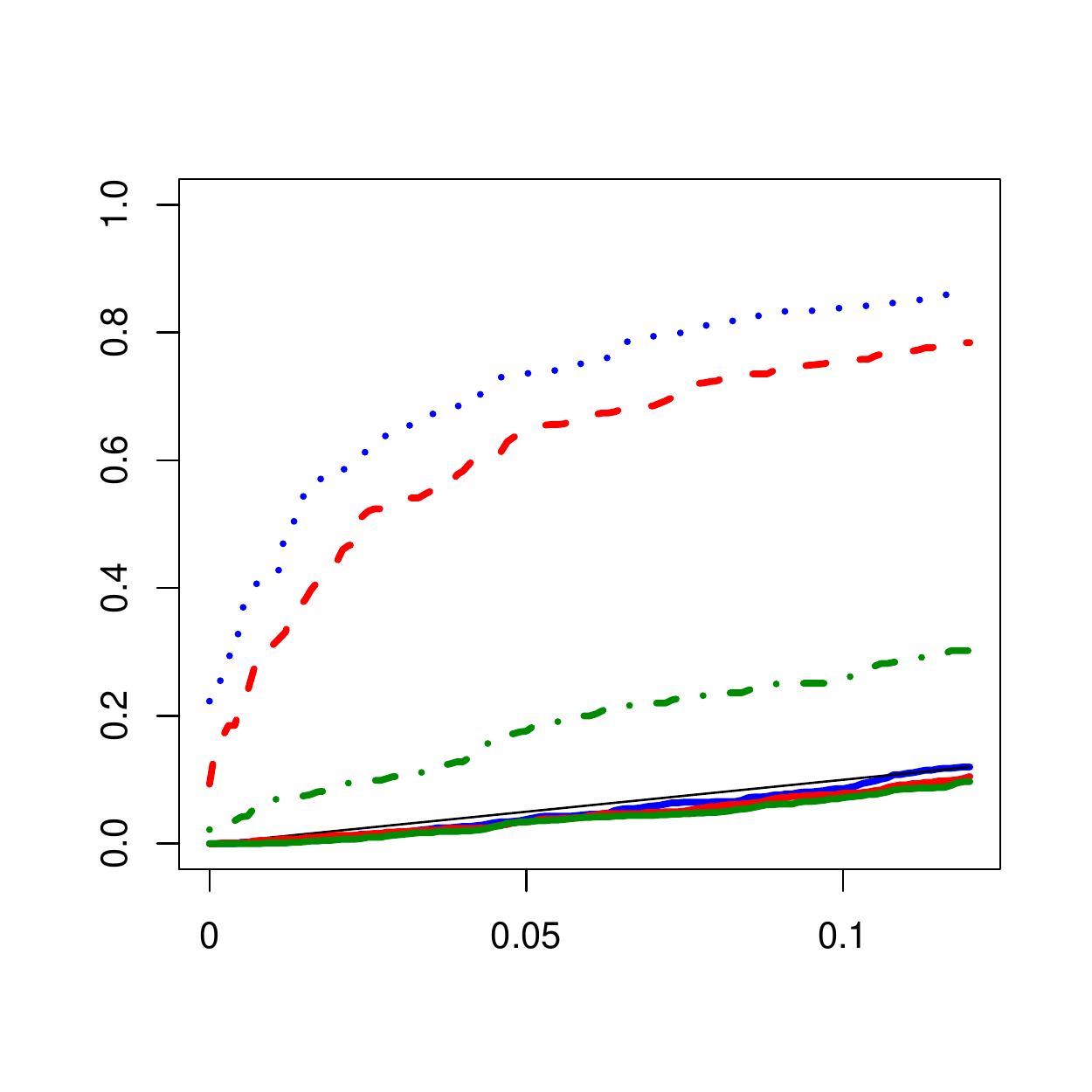}

\end{subfigure}
\hfill
\begin{subfigure}[t]{0.32\textwidth}
	\centering
	\includegraphics[width=\textwidth]{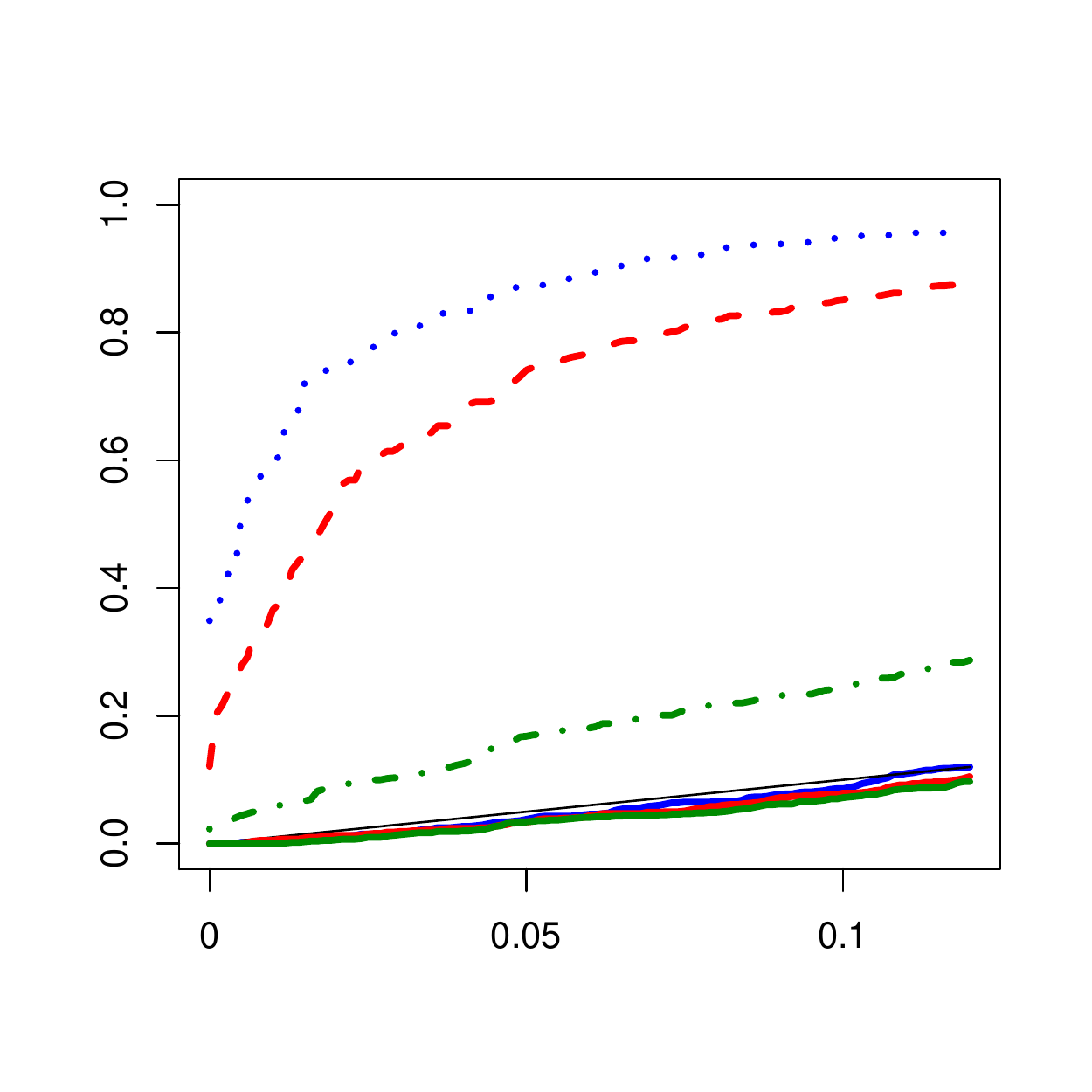}

\end{subfigure}
\begin{subfigure}[t]{0.32\textwidth}
\centering
	\includegraphics[width=\textwidth]{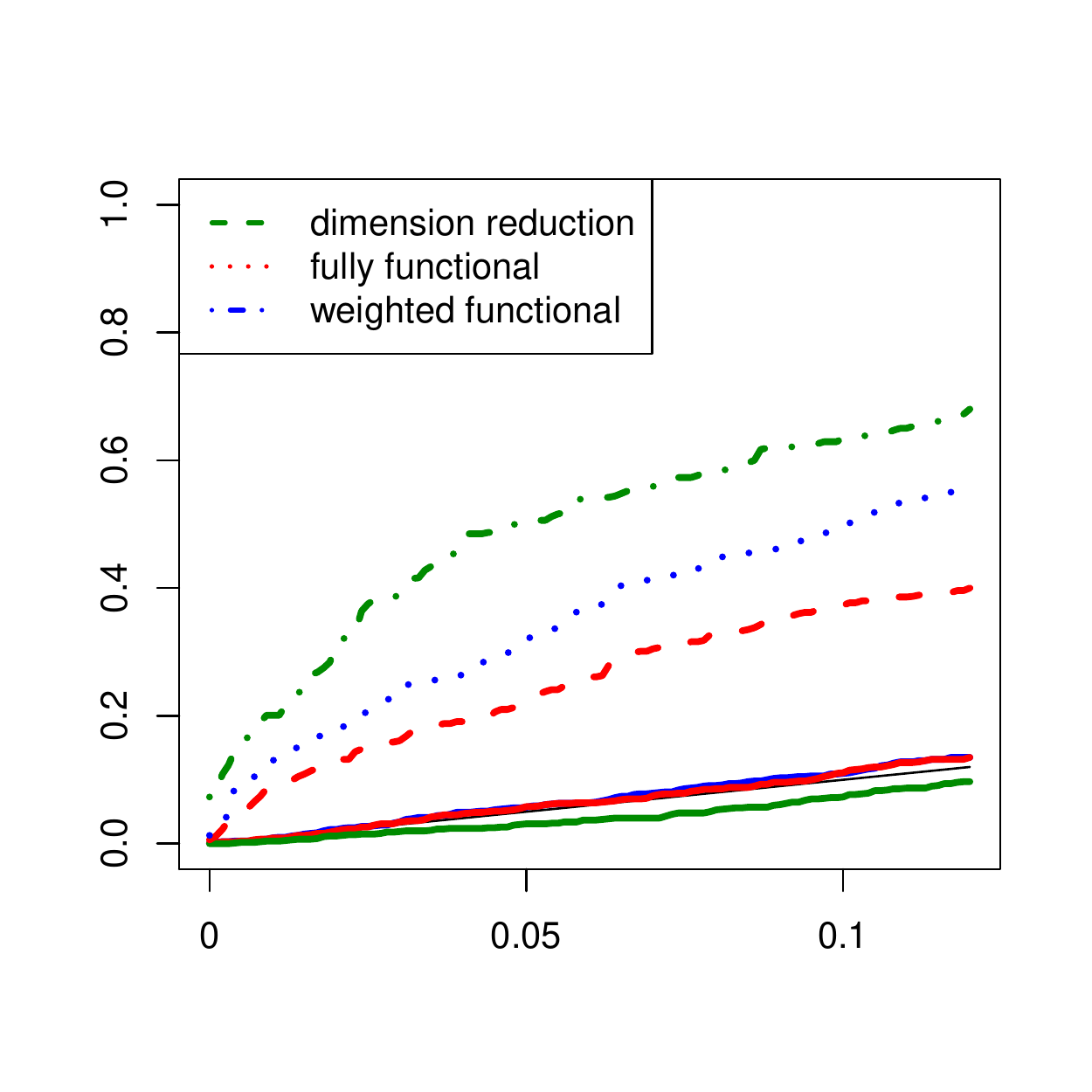}
	\caption{Setting 2, $\sigma_{\epsilon}=0.2,m=2,25,50$}
	\end{subfigure}
\hfill
\begin{subfigure}[t]{0.32\textwidth}
	\centering
	\includegraphics[width=\textwidth]{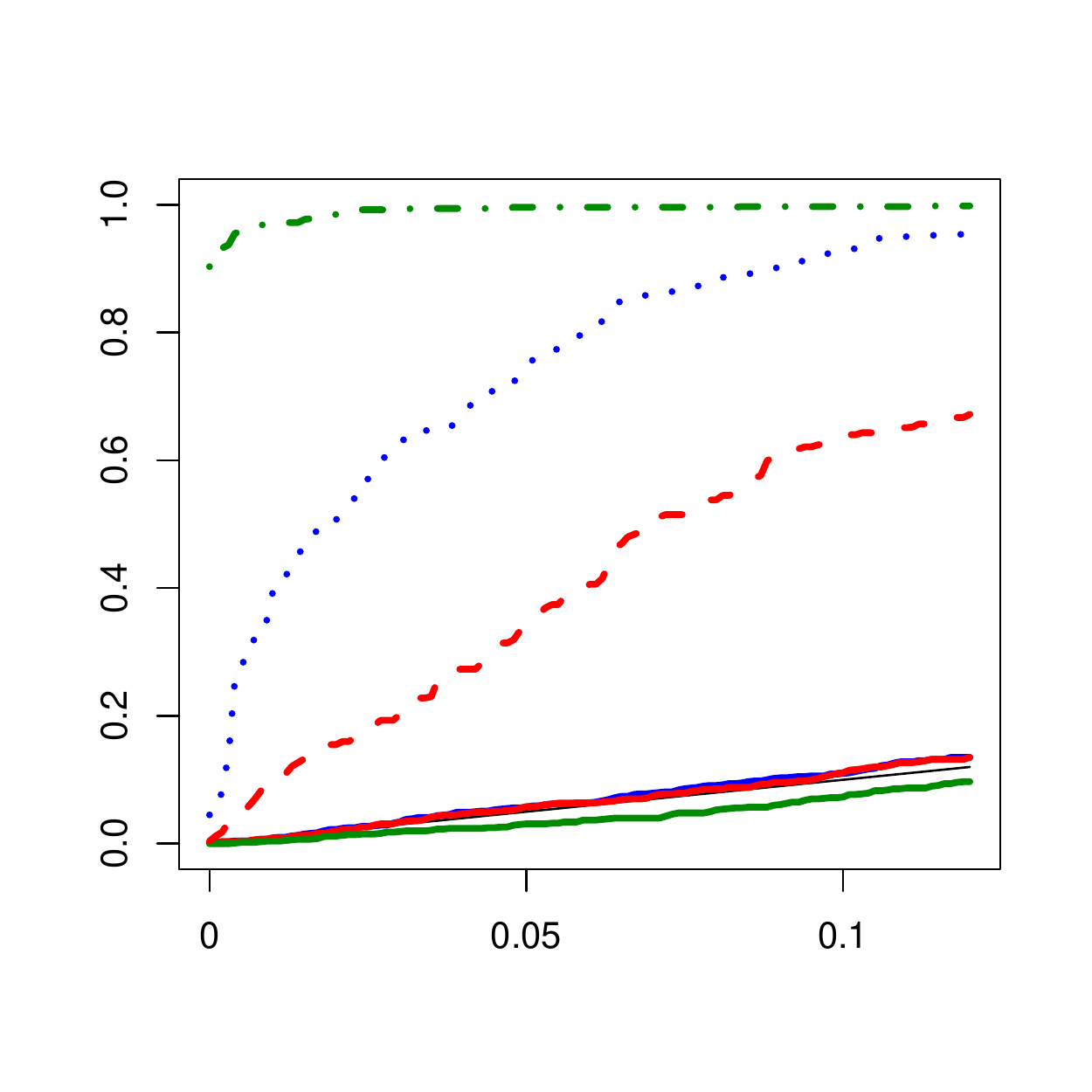}

\end{subfigure}
\hfill
\begin{subfigure}[t]{0.32\textwidth}
	\centering
	\includegraphics[width=\textwidth]{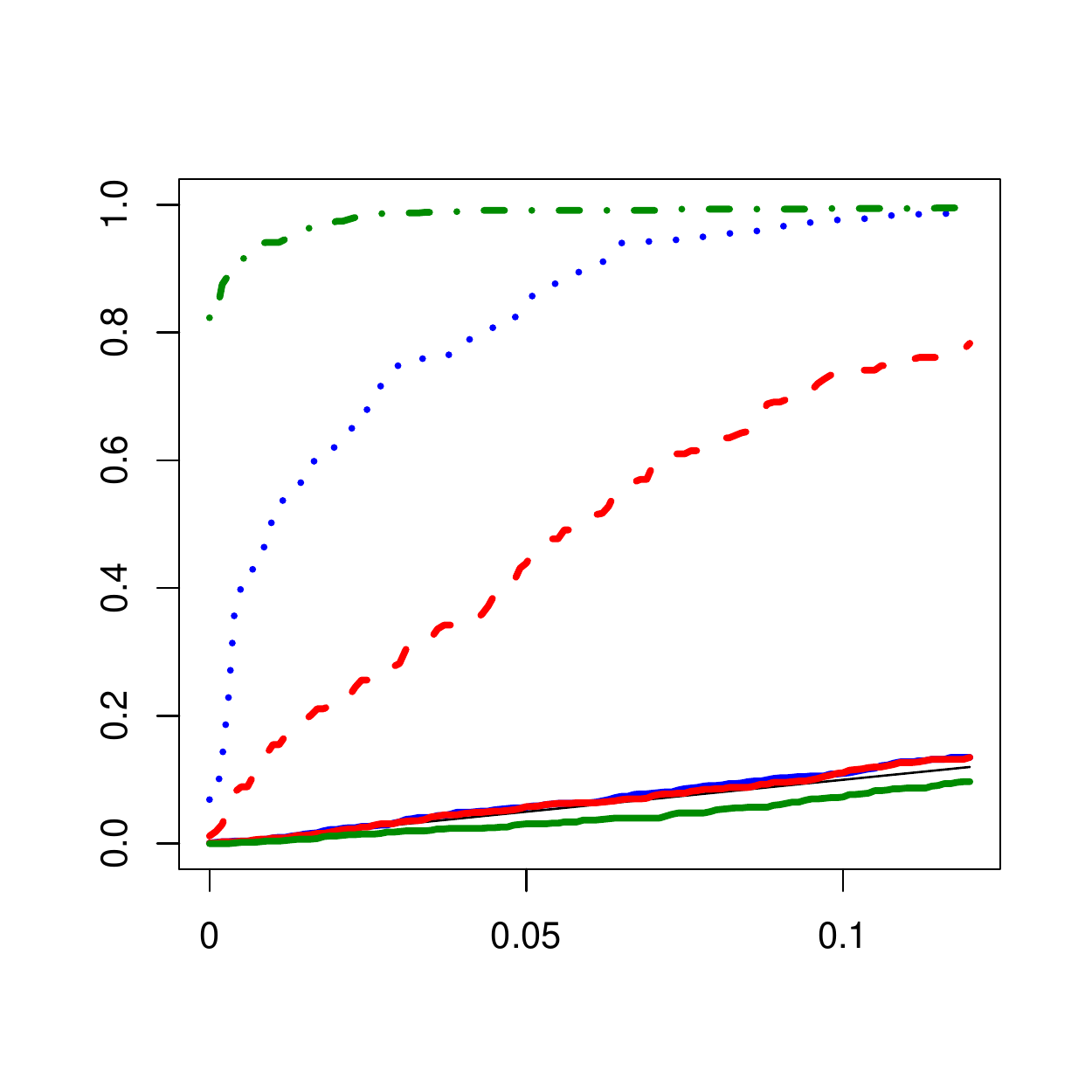}

\end{subfigure}
\begin{subfigure}[t]{0.32\textwidth}
\centering
	\includegraphics[width=\textwidth]{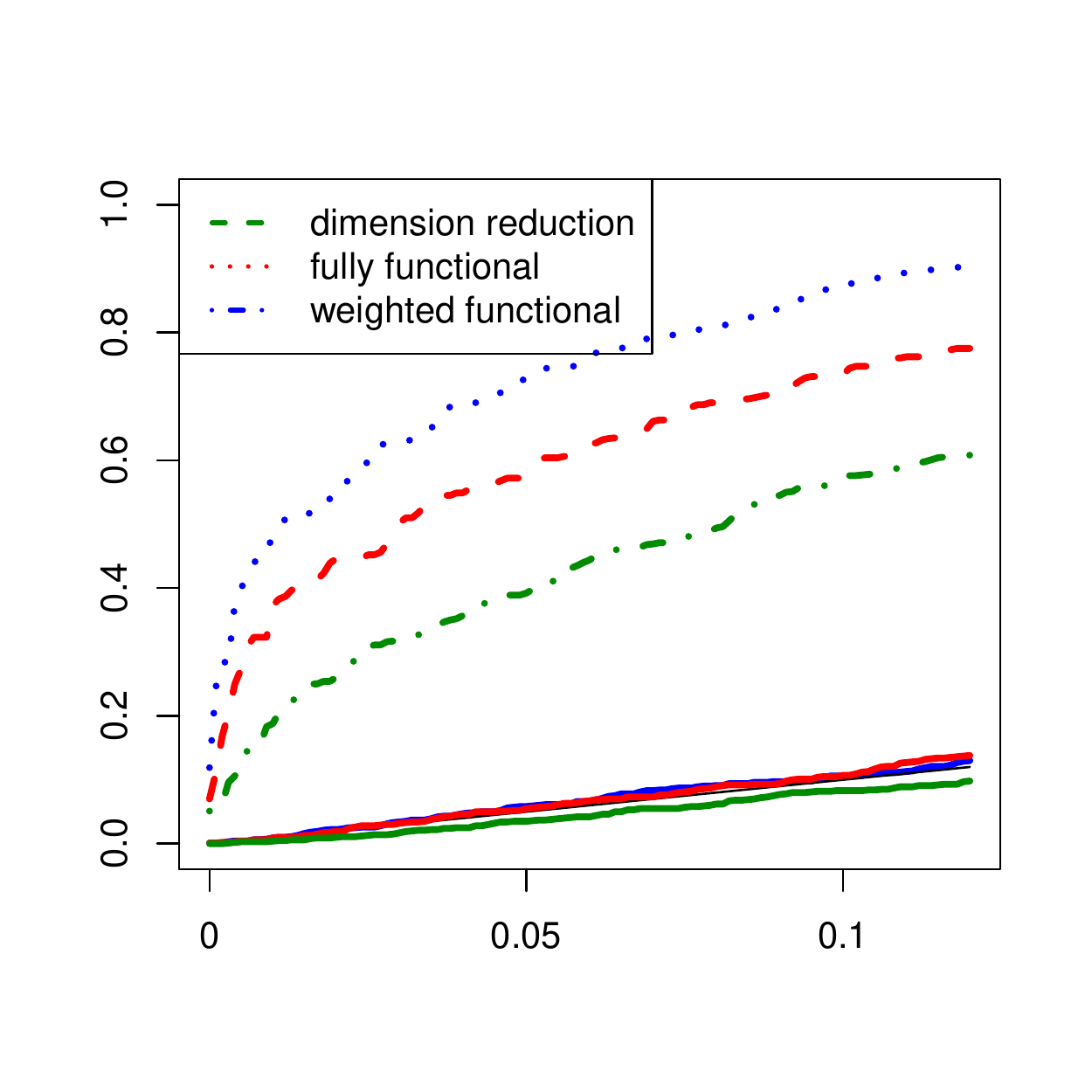}
	\caption{Setting 3, $\sigma_{\epsilon}=0.8,m=2,25,50$}
	\end{subfigure}
\hfill
\begin{subfigure}[t]{0.32\textwidth}
	\centering
	\includegraphics[width=\textwidth]{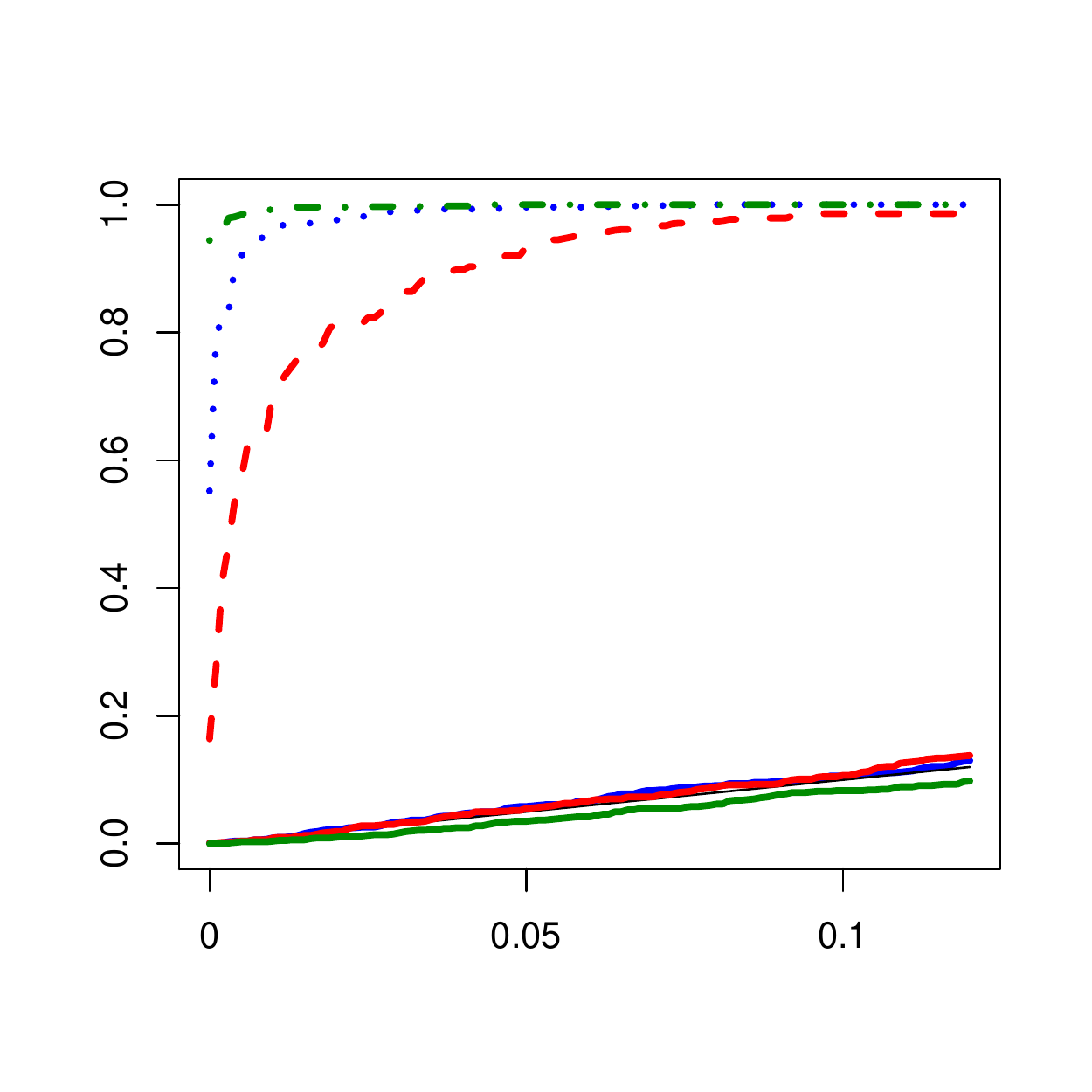}

\end{subfigure}
\hfill
\begin{subfigure}[t]{0.32\textwidth}
	\centering
	\includegraphics[width=\textwidth]{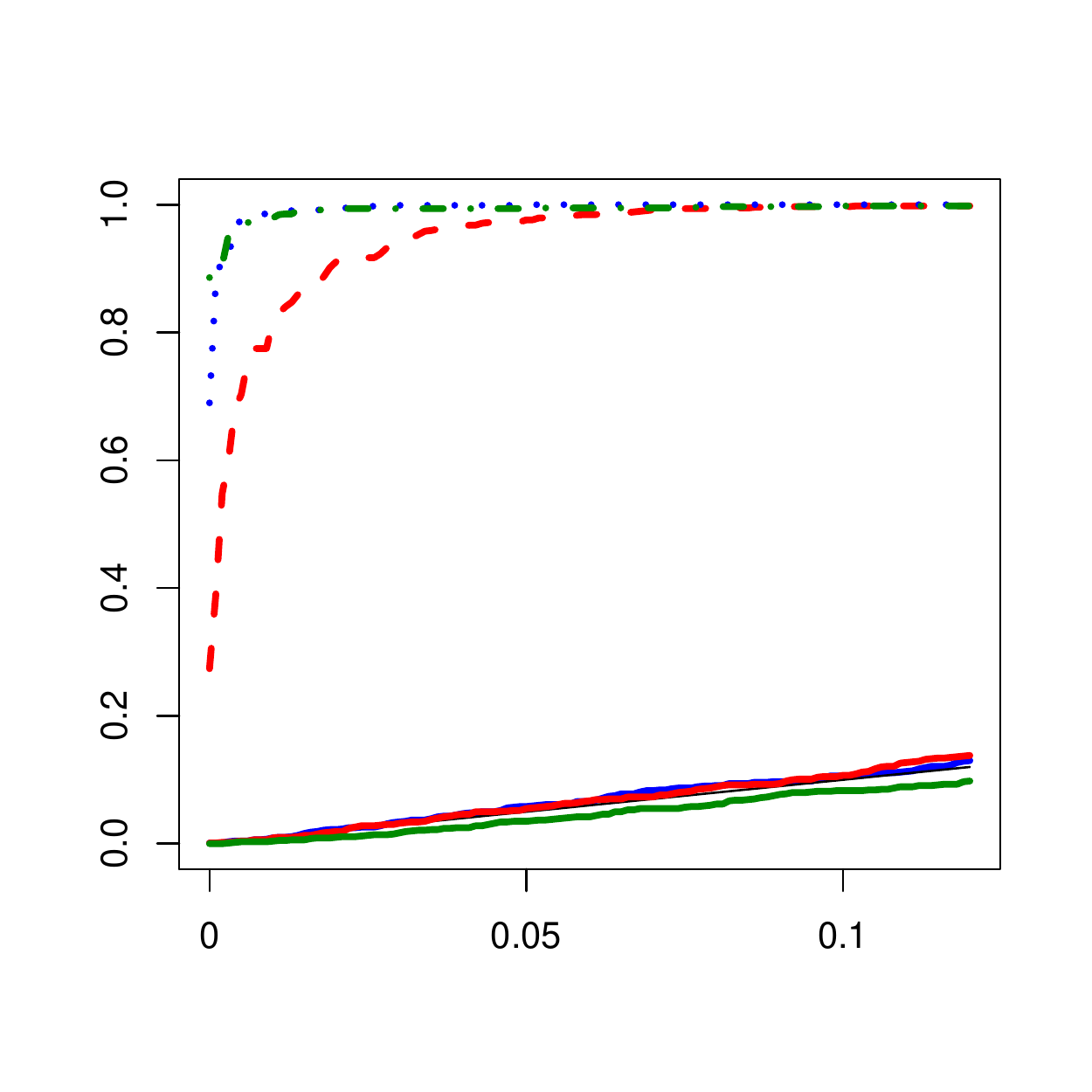}
\end{subfigure}
\caption{Empirical size (solid lines) and size corrected power (dashed lines) of the proposed procedures for independent data using the multivariate procedure after dimension reduction based on \eqref{tkd} (green), the fully functional procedure based on \eqref{tkf} (red) and the weighted functional procedure based on \eqref{tkw} (blue).}
\label{fig.sim}
\end{figure}
\section{Further results of the data analysis}\label{results.app}
In this section, we give some additional results of the analysis of the \textit{1000 Connectome Resting State Data} in order to complement the main findings that are reported in Section \ref{results} in the main paper. In the following we will make some remarks on the comparison of the $p$-values for the different procedures. First, it should be mentioned that the $p$-values obtained by the two functional procedures are consistent, meaning that in most of the cases they imply the same test decision and if they lead to different test decisions at a certain level $\alpha$ the $p$-values are nevertheless of the same magnitude, i.e. for one procedure the $p$-value is slightly below $\alpha$ and for the other procedure it slightly exceeds $\alpha$.  Regarding the comparison of the multivariate procedure with the functional procedures we observed that they lead to different test decisions in some cases. On the one hand, the multivariate procedure is not able to detect changes which are orthogonal to the projection subspace. On the other hand, false alarms can occur as the few components which are considered after reducing the dimension might contain some irregularities which lead to a rejection of the null hypothesis but are not significant when considering the full functional structure. This can be observed, for example, when analyzing sub34943. The multivariate procedure detects a deviation from covariance stationarity in the 8-dimensional time series of the scores but the null hypothesis is not rejected by the functional procedures. Figure \ref{red.example1} shows the 36 score products which are considered in the multivariate procedure. Calculating the componentwise $p$-values of the weighted functional procedure which includes 1083 score products for $\epsilon=0.0025$ it turns out that more than one third of the 36 components considered in the multivariate procedure belong to the 100 smallest $p$-values of the weighted functional procedure.
\begin{figure}[ht]
\centering
	\includegraphics[width=0.8\textwidth]{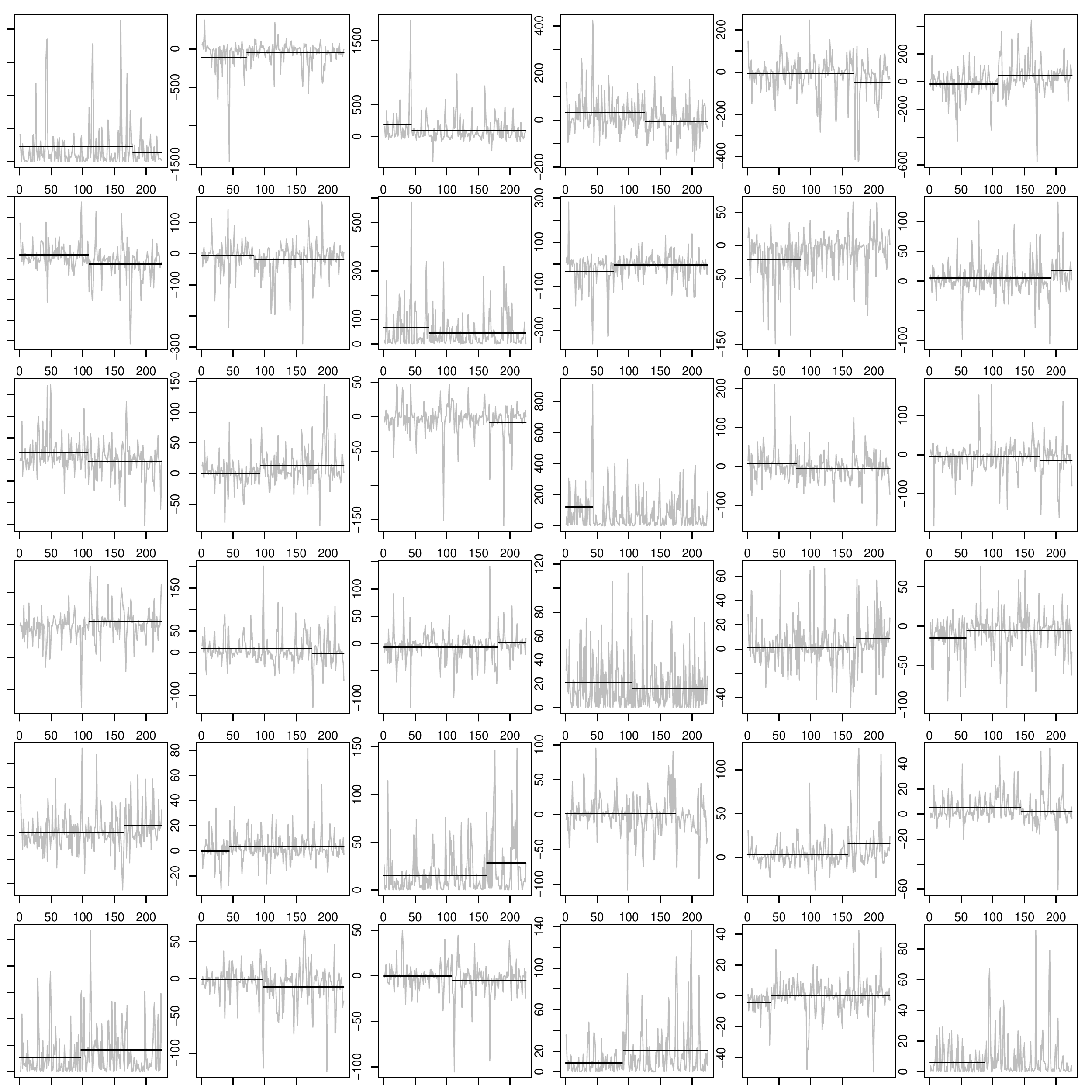}
	\caption{sub34943: All 36 score products obtained by dimension reduction.}
	\label{red.example1}
\end{figure}
In Section \ref{results} we have seen that for some data sets, as for example sub08816, the epidemic alternative is more appropriate. In addition to Figure \ref{ep.example2} in the main paper, Figure \ref{ep.example1} shows the 64 most significant score products for the AMOC alternative. A visual inspection of those two figures suggests that the epidemic model is indeed more suitable in this case. Furthermore, the small $p$-values are reasoned by the fact that the epidemic changes in the single components tend to be aligned.
\begin{figure}[ht]
\centering
	\includegraphics[width=0.8\textwidth]{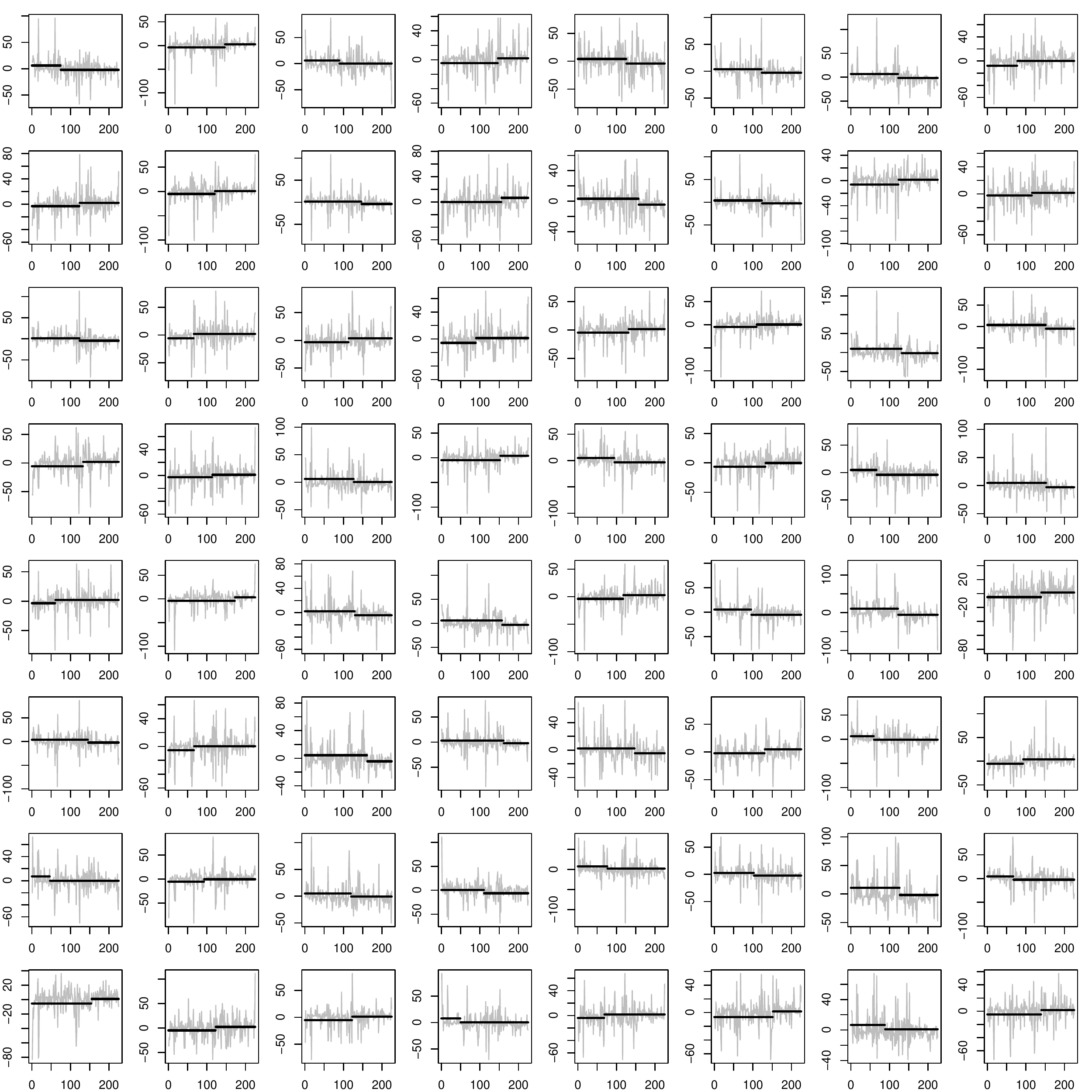}
	\caption{sub08816: 64 score products with the smallest $p$-values for the weighted functional statistic when testing for the AMOC alternative.}
	\label{ep.example1}
\end{figure}

\begin{figure}[ht]
\centering
	\includegraphics[width=0.8\textwidth]{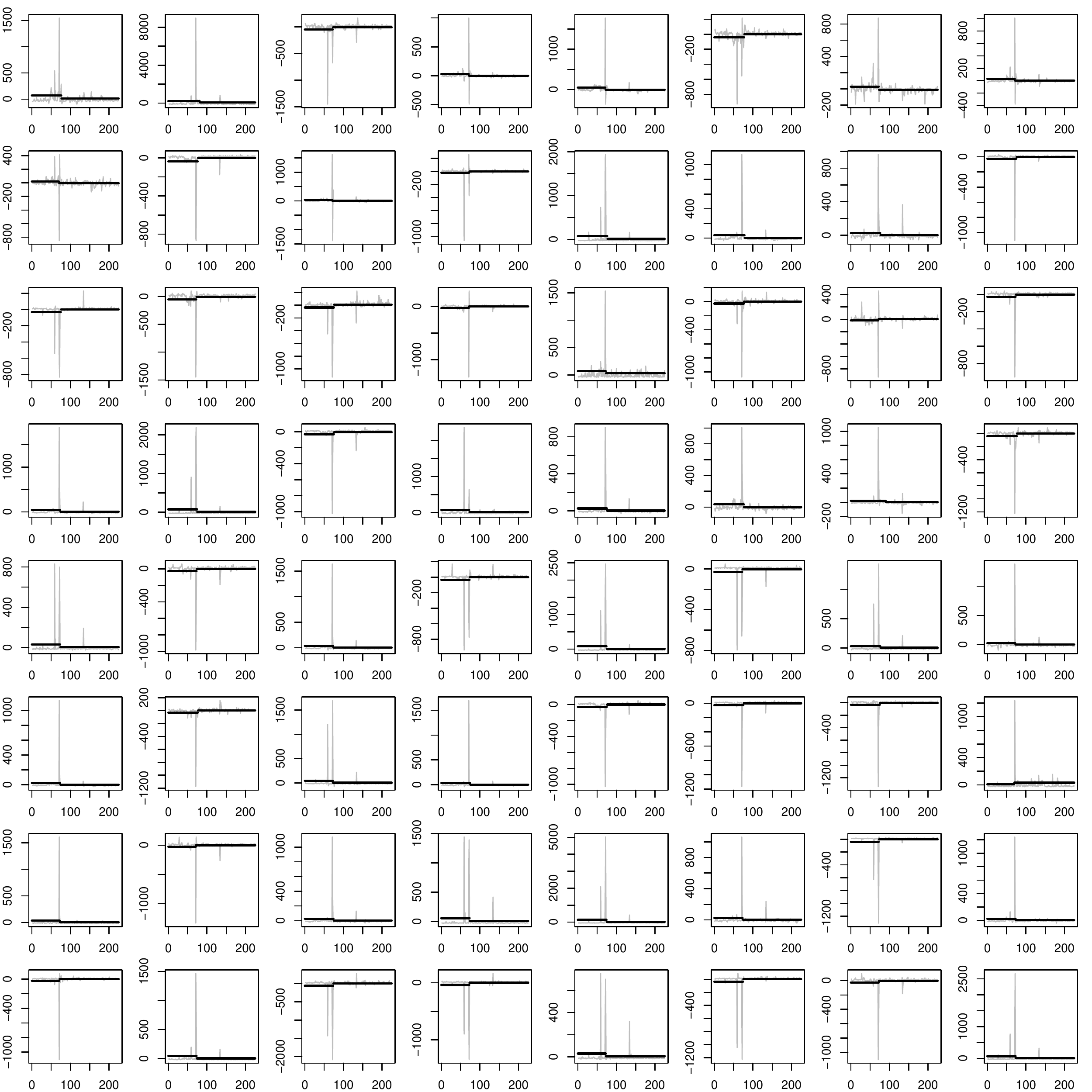}
	\caption{sub08992: 64 score products with the smallest $p$-values for the weighted functional statistic when testing for the AMOC alternative.}
	\label{out.example1}
\end{figure}
\begin{figure}[ht]
\centering
	\includegraphics[width=0.8\textwidth]{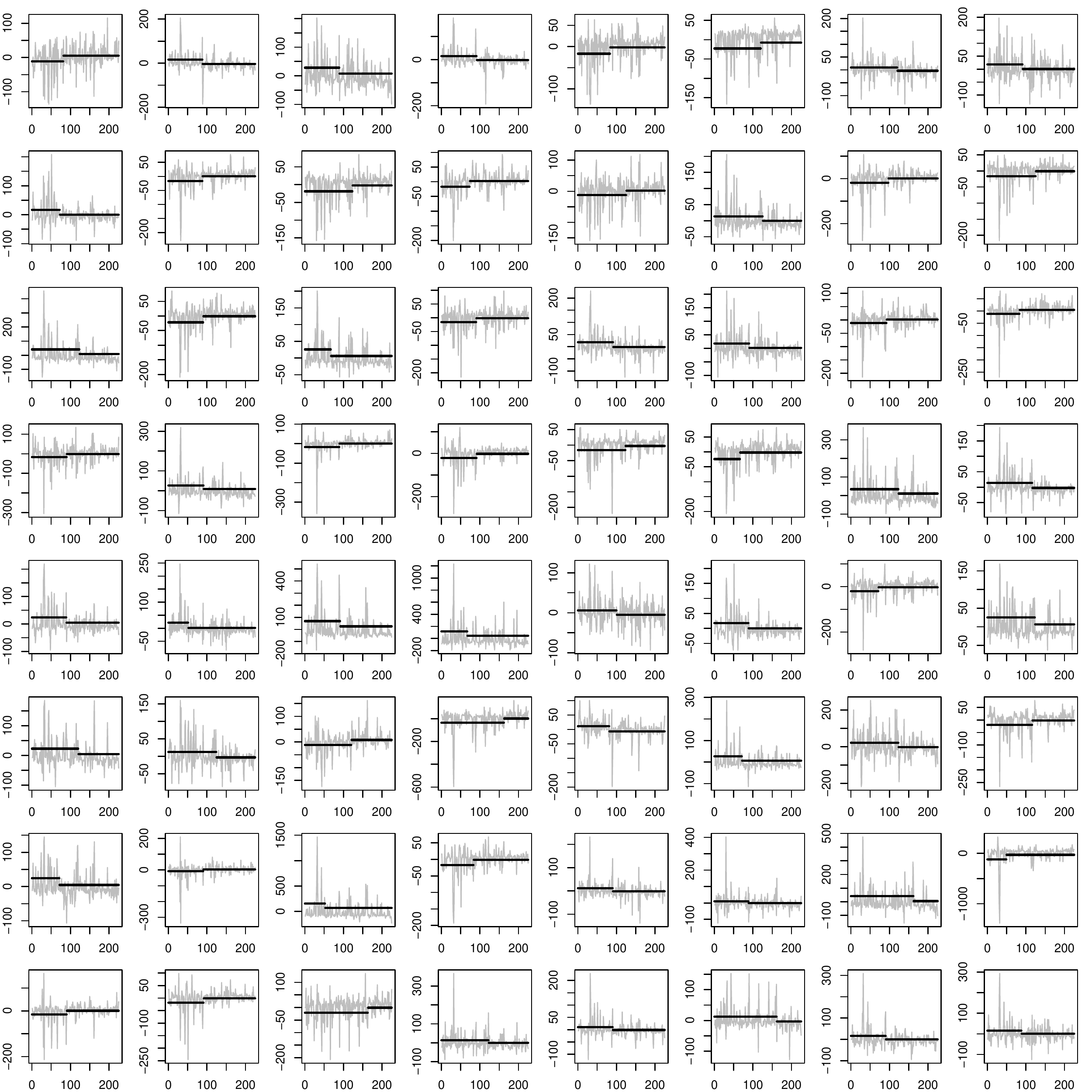}
	\caption{sub08455: 64 score products with the smallest $p$-values for the weighted functional statistic when testing for the AMOC alternative.}
	\label{ex08455}
\end{figure}

\end{document}